\documentclass{lmcs}
\pdfoutput=1
\usepackage[utf8]{inputenc}

\usepackage{lastpage}
\lmcsdoi{22}{1}{14}
\lmcsheading{}{\pageref{LastPage}}{}{}%
{May~16,~2024}{Feb.~27,~2026}{}

\usepackage{amsmath,systeme,amsthm,environ,amssymb}
\usepackage{amsfonts}
\usepackage{xspace}
\usepackage{float}
\usepackage{multicol}
\usepackage{soul}
\usetikzlibrary{positioning,tikzmark,arrows,scopes,backgrounds,automata,calc,matrix}
\pgfdeclarelayer{background}
\pgfdeclarelayer{foreground}
\pgfdeclarelayer{edgelayer}
\pgfdeclarelayer{nodelayer}
\pgfsetlayers{background,edgelayer,nodelayer,main,foreground}
\usepackage{tikz-cd}
\usepackage{pdflscape}
\usepackage{accents}
\usepackage{slantsc}
\usepackage{xstring}
\usepackage[most]{tcolorbox}
\usepackage{lipsum}
\usepackage{colonequals}
\usepackage{scalefnt}
\usepackage{caption}
\usepackage{subcaption}
\AddToHook{begindocument/before}{\usepackage{cleveref}%
\crefname{thm}{Theorem}{Theorems}
\crefname{cor}{Corollary}{Corollaries}
\crefname{lem}{Lemma}{Lemmas}
\crefname{prop}{Proposition}{Propositions}
\crefname{asm}{Asumption}{Asumptions}
\crefname{defi}{Definition}{Definitions}
\crefname{rem}{Remark}{Remarks}
\crefname{rems}{Remarks}{Remarks}
\crefname{exa}{Example}{Examples}
\crefname{exas}{Examples}{Examples}
\crefname{conj}{Conjecture}{Conjectures}
\crefname{prob}{Problem}{Problems}
\crefname{oprob}{Open Problem}{Open Problems}
\crefname{algo}{Algorithm}{Algorithms}
\crefname{obs}{Observation}{Observations}
\crefname{cond}{Condition}{Conditions}
\crefname{fact}{Fact}{Facts}
\crefformat{equation}{(#2#1#3)}
\crefrangeformat{equation}{(#3#1#4) to~(#5#2#6)}
\crefmultiformat{equation}{(#2#1#3)}%
{ and~(#2#1#3)}{, (#2#1#3)}{ and~(#2#1#3)}
}

\NewDocumentCommand{\ignore}{ m }{}

\newcommand{\mysf}[1]{{\textit{#1}}}

\setulcolor{lipicsYellowLight}
\setul{1.5pt}{1.0pt}

\definecolor{darkgray}{rgb}{0.31,0.31,0.33}
\definecolor[named]{lipicsGray}{rgb}{0.31,0.31,0.33}
\definecolor[named]{lipicsBulletGray}{rgb}{0.60,0.60,0.61}
\definecolor[named]{lipicsLineGray}{rgb}{0.51,0.50,0.52}
\definecolor[named]{lipicsLightGray}{rgb}{0.85,0.85,0.86}
\definecolor[named]{lipicsYellow}{rgb}{0.99,0.78,0.07}
\definecolor{cA}{HTML}{ED220D}
\definecolor{cB}{HTML}{F5C847}
\definecolor{cC}{HTML}{60D937}
\definecolor{cD}{HTML}{00A1FF}
\definecolor{cDark}{HTML}{333333}
\definecolor{lipicsLightGrayCopy}{rgb}{0.85,0.85,0.86}
\colorlet{cGrayLight}{lipicsLightGrayCopy!30!white}
\colorlet{cGrayDark}{lipicsLightGrayCopy!30!black}
\colorlet{cALight}{cA!35!white}
\colorlet{cBLight}{cB!35!white}
\colorlet{cCLight}{cC!35!white}
\colorlet{cDLight}{cD!35!white}
\colorlet{cADark}{cA!75!black}
\colorlet{cBDark}{cB!75!black}
\colorlet{cCDark}{cC!75!black}
\colorlet{cDDark}{cD!75!black}
\colorlet{cADarker}{cA!70!black}
\colorlet{cBDarker}{cB!70!black}
\colorlet{cCDarker}{cC!70!black}
\colorlet{cDDarker}{cD!70!black}
\definecolor{lipicsYellowCopy}{rgb}{0.99,0.78,0.07}
\colorlet{lipicsYellowLight}{lipicsYellowCopy!50!white}
\colorlet{lipicsYellowLighter}{lipicsYellowCopy!30!white}

\newif\ifdebuggingliterals
\debuggingliteralsfalse

\definecolor{literal-color}{HTML}{109905}
\definecolor{function-color}{HTML}{00786a}
\definecolor{object-color}{HTML}{78004e}
\definecolor{automata-class-color}{HTML}{785e00}

\NewDocumentCommand{\mathExpr}{ m }
{#1}
\NewDocumentCommand{\mathFunction}{ m O{} O{} O{a} m }%
{\mathExpr{\normalfont\mysf{\ifdebuggingliterals\color{function-color}\fi#1}_{#2}^{#3}
    \if\relax\detokenize{#5}\relax\else{\ifx#4a\bracketInner*{#5}\else\bracketInner[#4]{#5}\fi}\fi%
}}

\NewDocumentCommand{\mathObject}{ m O{} O{} O{a} m }%
{\mathExpr{\normalfont\textsc{\textsl{{{\ifdebuggingliterals\color{object-color}\fi#1}}}}_{#2}^{#3}\if\relax\detokenize{#5}\else{\ifx#4a\bracketInner*{#5}\else\bracketInner[#4]{#5}\fi}\fi}}

\newcommand{\suchthat}{\suchthatSymbol\PackageWarning{RadekTypesetting}{Command suchthat used outside of matching PairedDelimiter was used on input line \the\inputlineno.}}
\newcommand\suchthatSymbol[1][]{\nonscript\:#1\vert\allowbreak\nonscript\:\mathopen{}}

\renewcommand{\vec}[1]{\boldsymbol{#1}}

\DeclareSymbolFont{stmry}{U}{stmry}{m}{n}
\DeclareMathSymbol\shortrightarrow\mathrel{stmry}{"01}
\DeclareMathDelimiter\llbracket{\mathopen}{stmry}{"4A}
{stmry}{"71}
\DeclareMathDelimiter\rrbracket{\mathclose}{stmry}{"4B}
{stmry}{"79}

\DeclareFontFamily{OT1}{lmr}{}
\DeclareFontShape{OT1}{lmr}{m}{scsl}%
{<-> rm-lmcsco10}{}
\DeclareFontShape{OT1}{cmr}{m}{scit}{<->ssub*lmr/m/scsl}{}
\DeclareFontShape{OT1}{cmr}{m}{scsl}{<->ssub*lmr/m/scsl}{}

\DeclareFontFamily{OT1}{lmtt}{}
\DeclareFontShape{OT1}{lmtt}{sb}{n}
     {<-8.5>   rm-lmtt8     <8.5-9.5> rm-lmtt9
      <9.5-11> rm-lmtt10    <11->     rm-lmtt12
      }{}
\DeclareFontShape{OT1}{lmtt}{b}{n}
     {<-> rm-lmtk10}{}
\DeclareFontShape{OT1}{lmtt}{bx}{n}
     {<->ssub*lmtt/b/n}{}
\DeclareFontFamily{TS1}{cmtt}{}
\DeclareFontShape{TS1}{cmtt}{bx}{n}{<->ssub*lmtt/b/n}{}

\DeclarePairedDelimiterX{\setInner}[1]\{\}{\renewcommand\suchthat{\suchthatSymbol[\delimsize]}#1}
\NewDocumentCommand{\set}{ O{a} m }{\ifx#1a\setInner*{#2}\else\setInner[#1]{#2}\fi}
\NewDocumentCommand{\bracket}{ O{a} m }{\ifx#1a\bracketInner*{#2}\else\bracketInner[#1]{#2}\fi}

\DeclarePairedDelimiter {\size}          \lvert\rvert
\DeclarePairedDelimiter {\myMatrix}      []
\DeclarePairedDelimiter {\squareBracket}      []
\DeclarePairedDelimiter {\sem}           \llbracket\rrbracket
\DeclarePairedDelimiter {\tuple}         ()
\DeclarePairedDelimiter {\bracketInner}  ()

\newcommand{\hiddenset}[1]{#1}

\protected\def\verythinspace{%
    \ifmmode\mskip0.5\thinmuskip\else\ifhmode\kern0.08334em\fi\fi%
}

\newcommand{\argumentDot}{\verythinspace\cdot\verythinspace}

\protected\def\verythinspace{%
    \ifmmode\mskip0.5\thinmuskip\else\ifhmode\kern0.08334em\fi\fi%
}

\let\forallSymbol=\forall
\let\existsSymbol=\exists
\let\limplies=\rightarrow

\RenewDocumentCommand{\forall}{m t.}{%
        {{\forallSymbol}\verythinspace\IfValueT{#1}{#1\IfBooleanT{#2}{.\,}\;}}%
}
\RenewDocumentCommand{\exists}{m t.}{%
        {{\existsSymbol}\verythinspace\IfValueT{#1}{#1\IfBooleanT{#2}{.\,}\;}}%
}

\NewDocumentCommand{\lang}{}{\mathFunction{\ensuremath{{\mathcal{L}}}}}
\NewDocumentCommand{\langUniv}{ O{} }{\mathFunction{${\mathcal{L}}$}[#1][\forall\!]}
\NewDocumentCommand{\langForall}{ O{\I} }{\mathFunction{$\rho$}[#1][\forallSymbol\!]}
\NewDocumentCommand{\rel}{}{\mathFunction{\ensuremath{{\mathcal{R}}}}}

\NewDocumentCommand{\N}{}{\bN}
\NewDocumentCommand{\I}{}{\cI}
\NewDocumentCommand{\F}{}{\mathExpr{\Phi}}

\NewDocumentCommand{\colours}{}{\mathFunction{colours}}
\NewDocumentCommand{\tTop}{}{\mathFunction{top}}
\NewDocumentCommand{\tRight}{}{\mathFunction{right}}
\NewDocumentCommand{\tBottom}{}{\mathFunction{bottom}}
\NewDocumentCommand{\tLeft}{}{\mathFunction{left}}

\NewDocumentCommand{\BigO}{}{\mathFunction{$O$}}

\NewDocumentCommand{\tilingProblem}{}{\mathObject{CorridorTiling}{}}

\NewDocumentCommand{\enc}{ O{\I} }{\mathFunction{enc}[#1]}
\NewDocumentCommand{\encRow}{ O{\I} }{\mathFunction{encRow}[#1][]}
\NewDocumentCommand{\encCell}{ O{\I} }{\mathFunction{encCell}[#1][]}

\NewDocumentCommand{\regexToAutomaton}{ }{\mathFunction{$\mathcal{A}$}}
\NewDocumentCommand{\automatonToRelation}{ m }{\langToRel{\lang{#1}}}
\NewDocumentCommand{\relToLang}{ }{\mathFunction{Rel2Lang}}
\NewDocumentCommand{\langToRel}{ }{\mathFunction{Lang2Rel}}
\NewDocumentCommand{\encLang}{ O{\I} }{\mathExpr{L_{#1}}}
\NewDocumentCommand{\valEncLang}{ O{\I} }{\mathObject{ValidEnc}[#1]{}}
\NewDocumentCommand{\condLang}{ m O{\I} }{\mathObject{Cond}[#2][#1]{}}
\NewDocumentCommand{\allSufAut}{ }{\mathObject{AllSuf}}
\NewDocumentCommand{\allSufLang}{ }{\mathObject{$L$}[\forallSymbol\text{suf}]}

\NewDocumentCommand{\rhoProj}{ O{a} m }{\mathFunction{$\rho$}[\I][][#1]{#2}}
\NewDocumentCommand{\rhoInv}{ O{a} m }{\mathFunction{$\rho$}[\I][-1][#1]{#2}}
\NewDocumentCommand{\inProj}{ O{a} m }{\mathFunction{$\psi$}[\inSub][][#1]{#2}}
\NewDocumentCommand{\outProj}{ O{a} m }{\mathFunction{$\psi$}[\outSub][][#1]{#2}}
\NewDocumentCommand{\proj}{ O{} O{a} m }{\mathFunction{$\pi$}[#1][\existsSymbol][#2]{#3}}
\NewDocumentCommand{\univProj}{ O{} O{a} m }{\mathFunction{$\pi$}[#1][\forallSymbol][#2]{#3}}

\NewDocumentCommand{\biglanguagesum}{}{\sum}
\NewDocumentCommand{\biglanguageconcat}{}{\prod}
\NewDocumentCommand{\setComplement}{ m }{\overline{#1}}
\NewDocumentCommand{\concat}{}{\,}
\NewDocumentCommand{\languagesum}{}{+}

\makeatletter
\newcommand*{\coloncoloneqqq}{%
\ensuremath{%
    \mathrel{%
        \@center@colon
        \colonsep
        \@center@colon
        \doublecolonsep
        {\equiv}
    }%
}%
}
\makeatother

\makeatletter
\newcommand{\fixed@sra}{$\vrule height 2\fontdimen22\textfont2 width 0pt\shortrightarrow$}
\newcommand{\shortarrow}[1]{%
    \mathrel{\text{\rotatebox[origin=c]{\numexpr#1}{\fixed@sra}}}
}
\makeatother

\newcommand{\toplefttile}{t^{\kern-2pt\shortarrow{120}\kern-1pt}}
\newcommand{\bottomrighttile}{t^{\kern-2pt\shortarrow{300}\kern-1pt}}

\ExplSyntaxOn
\DeclareExpandableDocumentCommand{\IfNoValueOrEmptyTF}{mmm}
 {
  \IfNoValueTF{#1}{#2}
   {
    \tl_if_empty:nTF {#1} {#2} {#3}
   }
 }
\ExplSyntaxOff

\NewDocumentCommand{\NewPredicate}{m m E{^_}{{}{}} }%
{%
    \expandafter\NewDocumentCommand\expandafter{\csname p#1\endcsname}%
    {O{a} E{^_}{{#3}{#4}} m }%
    {\mathExpr{\IfNoValueOrEmptyTF{##2}{{\mathrm{#2}}_{##3}}{{\mathrm{#2}}^{##2}_{##3}}\if\relax\detokenize{##4}\relax\else{\ifx##1a\bracketInner*{##4}\else\ifx##1b({##4})\else\bracketInner[##1]{##4}\fi\fi}\fi}}%
}
\NewDocumentCommand{\NewVariable}{m}%
{%
    \expandafter\NewDocumentCommand\expandafter{\csname v#1\endcsname}%
    {}{\mathExpr{{\mathit{#1}}}}%
}

\NewPredicate{V}{V}_2

\NewPredicate{P}{P}_2
\NewPredicate{Bit}{BitAt}
\NewPredicate{Tile}{Tile}_{\I}
\NewPredicate{Read}{NumAt}_{k}
\NewPredicate{TopLeft}{TopLeft}_{\I}
\NewPredicate{BottomRight}{BottomRight}_{\I}
\NewPredicate{MatchH}{MatchH}_{\I}
\NewPredicate{MatchV}{MatchV}_{\I}
\NewPredicate{MaxNum}{MaxNum}_{k}
\NewPredicate{ImplicitBit}{ImplicitBitLit}_{n}
\NewPredicate{WidthMul}{WidthMul}_{n}
\NewPredicate{Width}{Width}_{n}
\NewPredicate{IsARuler}{IsARuler}_{n}
\NewPredicate{ETiling}{ETiling}_{\I}
\NewPredicate{ETilingPrime}{ETiling}^{\prime}_{\I}
\NewPredicate{bin}{bin}
\NewPredicate{Mul}{Mul}_{3^n}
\NewPredicate{Mod}{LogEqMod}_{W_n}

\NewVariable{pre} 
\NewVariable{suf}  

\newcommand{\abbrev}[1]{#1\@\xspace}

\newcommand{\eg}  {e.g.\@\xspace}

\newcommand{\ie}  {i.e.\@\xspace}

\newcommand{\Buchi}{Büchi\xspace}

\newcommand{\proofCase}[1]{\par\smallskip\noindent{\mbox{\textsc{{#1.}}}}\ }

\NewDocumentCommand{\inSub}{}{\text{in}}
\NewDocumentCommand{\outSub}{}{\text{out}}
\NewDocumentCommand{\ini}{}{\text{I}}
\NewDocumentCommand{\fin}{}{\text{F}}
\NewDocumentCommand{\finite}{}{\text{fin}}
\NewDocumentCommand{\increment}{}{\text{inc}}
\NewDocumentCommand{\inc}{}{\text{\tikzxmark}}
\NewDocumentCommand{\corrrr}{}{\text{\tikzcmark}}

\theoremstyle{definition}
\newtheorem{cond}{Condition}

\newcommand{\logspace}   {\textsc{LogSpace}\xspace}

\newcommand{\expexpspace}{\textsc{2-ExpSpace}\xspace}

\newcommand{\expspace}   {\textsc{ExpSpace}\xspace}

\newcommand{\NFA}       {{\normalfont\textsc{nfa}}\xspace}

\newcommand{\cLetter}[1]{\mathExpr{\mathcal{#1}}}
\newcommand{\cA}{\cLetter{A}}
\newcommand{\cB}{\cLetter{B}}
\newcommand{\cC}{\cLetter{C}}

\newcommand{\cE}{\cLetter{E}}
\newcommand{\cF}{\cLetter{F}}
\newcommand{\cG}{\cLetter{G}}
\newcommand{\cH}{\cLetter{H}}
\newcommand{\cI}{\cLetter{I}}

\newcommand{\cT}{\cLetter{T}}

\newcommand{\bLetter}[1]{\mathExpr{\mathbb{#1}}}

\newcommand{\bN}{\bLetter{N}}

\newcommand{\tLetter}[1]{\mathExpr{\mathtt{#1}}}

\newcommand{\tl}{\tLetter{l}}

\tikzset{
  thick/.style=      {line width=0.7pt},
}

\tikzstyle{my-line} = [
  draw=black, 
  thick,
  rounded corners=1mm
]
\tikzstyle{my-arrow} = [my-line,-{Triangle}]
\tikzstyle{my-arrow-both-ends} = [my-line, {Triangle}-{Triangle}]
\tikzstyle{shortened} = [shorten >=4pt,shorten <=4pt]
\tikzstyle{separated} = [shortened]

\tikzstyle{my-line-thin} = [
  my-line,
  semithick,
]

\newcommand{\myEdge}[4]{
  \ifthenelse{\equal{#3}{}}{
    \draw[my-line, #1] 
      (#2) to
      (#4);
  }{
    \draw[my-line, #1] 
      (#2) to 
        node[EdgeLabel] {{\strut\small#3}} 
      (#4);
  }
}

\tikzstyle{state} = [
circle,
inner sep=0,
minimum size=0.6cm,
my-line-thin,
]

\tikzstyle{my-state} = [
circle,
thick,
inner sep=0,
minimum size=0.2cm,
draw=black,
]

\tikzstyle{EdgeLabel} = [
  circle,
  line width=0.12cm,
  midway,
  auto,
  text height=1.1ex,
  text depth=0.2ex,
  text width=1em,
  align=center,
  fill=white,
  inner ysep=0.0cm,
  inner xsep=0.0cm,
  minimum size=0,
  opacity=1,
  text opacity=1
]

\tikzstyle{transition} = [
  ->,
  shorten >= 2pt,
  shorten <= 2pt
]

\definecolor{accentGray}{RGB}{71, 71, 75}
\definecolor{accent}{RGB}{252, 199, 18}
\tikzstyle{letter} = [inner sep=0.8mm, anchor=north west, yshift=-2mm]
\tikzstyle{bgLine} = [draw=accentGray, line width=0.2mm]
\tikzstyle{box} = [bgLine, fill=white, inner sep=0.8mm, anchor=north west, outer sep=0, text=black]
\tikzstyle{dot} = [fill=black, inner sep=0.8mm, outer sep=0]
\tikzstyle{axis} = [line width=0.5mm, draw=accent, -{Triangle[scale=1]}]

\newcommand{\tikzxmark}{%
  \tikz[scale=1.0,baseline=0.2ex] {
    \draw[line width=0.7,line cap=round] (0,0) to [bend left=6] (1ex,1ex);
    \draw[line width=0.7,line cap=round] (0.2ex,0.95ex) to [bend right=3] (0.8ex,0.05ex);
  }}
\newcommand{\tikzcmark}{%
  \tikz[scale=1.0,baseline=0.2ex] {
    \draw[line width=0.7,line cap=round] (0.25ex,0) to [bend left=10] (1ex,1ex);
    \draw[line width=0.8,line cap=round] (0,0.35ex) to [bend right=1] (0.23ex,0);
  }}

\makeatletter
\newcommand{\tikzmarkMath}[2]{%
  \tikz[%
    remember picture,
    baseline = (#1.base),
    inner sep = 0pt,
    outer sep = 0pt, 
  ] \node (#1) {$\m@th\displaystyle #2$};%
}
\makeatother
\tikzstyle{tile label}=[]

\newcommand{\tileInner}[4]{
  \begin{scope}[scale=0.34]
    \def\mycolours{{"cGrayLight","cDark","cBLight","cCLight","cALight","cDLight"}}
    \def\mytextcolours{{"black","white","black","black","black","black"}}
    \foreach \a/\col in {90/#1,0/#2,270/#3,180/#4} {
      \pgfmathsetmacro{\colour}{\mycolours[\col]}
      \pgfmathsetmacro{\textcolour}{\mytextcolours[\col]}
      \fill[\colour,rotate=\a] (0,0) -- (1,-1) -- (1,1) -- cycle;
      \path (0,0) +(\a:0.63) node[tile label,inner sep=0] {\scriptsize\color{\textcolour}$\col$};
    }
    \draw (1,1) -- (-1,-1) (-1,1) -- (1,-1) (1,1) rectangle (-1,-1);
  \end{scope}
}
\input{config/ABACABA-word-utils}


\begin{document}
    \title[Universal quantification makes automatic structures hard to decide]{Universal quantification makes automatic structures hard to decide}

    \author[C.\ Haase]{Christoph Haase\lmcsorcid{https://orcid.org/0000-0002-5452-936X}}[a]

    \author[R.\ Piórkowski]{Radosław Piórkowski\lmcsorcid{https://orcid.org/0000-0002-9643-182X}}[a,b]

    \address{Department of Computer Science, University of Oxford, Oxford, United Kingdom}	

    \address{Department of Computer Science, University of Warwick, Coventry, United Kingdom}	

    \begin{abstract}
      Automatic structures are first-order structures whose universe and relations
      can be represented as regular languages. It follows from the
      standard closure properties of regular languages that the
      first-order theory of an automatic structure is decidable. While
      existential quantifiers can be eliminated in linear time by
      application of a homomorphism, universal quantifiers are
      commonly eliminated via the identity $\forall{x}. \Phi
      \equiv \neg (\exists{x}. \neg \Phi)$. If $\Phi$ is
      represented in the standard way as an NFA, a priori this
      approach results in a doubly exponential blow-up. However, the
      recent literature has shown that there are classes of automatic
      structures for which universal quantifiers can be eliminated by
      different means without this blow-up by treating them as
      first-class citizens and not resorting to double
      complementation. While existing lower bounds for some classes of
      automatic structures show that a singly exponential blow-up is
      unavoidable when eliminating a universal quantifier, it is not
      known whether there may be better approaches that avoid the
      na\"ive doubly exponential blow-up, perhaps at least in
      restricted settings.

      In this paper, we answer this question negatively and show that
      there is a family of NFA representing automatic relations for
      which the minimal NFA recognising the language after eliminating
      a single universal quantifier is doubly exponential, and
      deciding whether this language is empty is \expspace-complete.

      The techniques underlying our \expspace lower bound
      further enable us to establish new lower bounds for
      some fragments of B\"uchi arithmetic with a fixed 
      number of quantifier alternations.
    \end{abstract}
    \typeout{HERE0}
    \maketitle

    \section{Introduction}
    
Quantifier elimination is a standard technique to decide logical
theories. A logical theory $\cT$ admits quantifier elimination whenever
for every quantifier free conjunction of literals
$\Phi(x,y_1,\ldots,y_n)$ of $\cT$ there is a quantifier free formula
$\Psi(y_1,\ldots,y_n)$ such that $\cT \models \exists{x}. \Phi
\leftrightarrow \Psi$. Universal quantifiers can then be eliminated
simply by applying the duality $\forall{x}. \Phi \equiv \neg (\exists{x}. \neg\Phi)$. If the formula $\Psi$ above is effectively computable
then $\cT$ is decidable provided that its quantifier-free fragment is decidable. 
For quantifier elimination procedures, the
computationally most expensive step is the elimination of an
existential quantifier, since negating a formula can be performed on a
syntactic level.

Automatic structures~\cite{Hod82,KN95,BG00} are a family of
first-order structures whose corresponding first-order theory can be
decided using automata-theoretic methods, as an alternative approach
to syntactic quantifier elimination. In their simplest variant,
automatic structures are relational first-order structures whose
universe is isomorphic to a regular language $L\subseteq \Sigma^*$
over some alphabet $\Sigma$, and whose $n$-ary relations are
interpreted as regular languages over $(\Sigma^n)^*$. It follows that
the set of all satisfying assignments of a quantifier-free formula
$\Phi(x_1,\ldots,x_{m+1})$ can be obtained as the language
$\lang{\cA}\subseteq (\Sigma^{m+1})^*$ of some finite-state automaton
$\cA$. In this setting, eliminating existential quantifiers is
easy. In order to obtain a finite-state automaton whose language
encodes the satisfying assignments to $\exists{x_{m+1}}. \Phi$,
it suffices to apply the homomorphism induced by the mapping $h\colon
(\Sigma^{m+1}) \to (\Sigma^{m})$ such that
$h(u_1,\ldots,u_{m+1})\coloneqq (u_1,\ldots,u_m)$ to $\lang{\cA}$. This
can be performed in linear time, even when $\cA$ is
non-deterministic. However, if $\cA$ is non-deterministic then
computing a finite-state automaton whose language encodes the
complement of $\Phi$ is computationally difficult and may lead to an
automaton with $2^{\Omega(\size \cA)}$ states. In particular, due
to double complementation, eliminating a universal quantifier
may \emph{a priori} lead to an automaton with $2^{2^{\Omega(\size{\cA)}}}$
states. Notable examples of automatic structures are Presburger
arithmetic~\cite{Pres29}, the first-order theory of the structure
$\langle \N,0,1,+,=\rangle$, and its extension B\"uchi
arithmetic~\cite{B60,Bru85,BHMV94}. Tool suites such
as \textsc{Lash}~\cite{LASH}, \textsc{Tapas}~\cite{LP09}
and \textsc{Walnut}~\cite{Mous16} are based on the automata-theoretic
approach and have successfully been used to decide challenging
instances of Presburger arithmetic and B\"uchi arithmetic from various
application domains. Those tools eliminate universal quantifiers via
double complementation.

Yet another approach to deciding Presburger arithmetic is based on
manipulating semi-linear sets~\cite{GS66,CHM22}, which are
generalisations of ultimately periodic sets to arbitrary tuples of
integers in $\N^d$. They are similar to automata-based methods in
terms of the computational difficulty of existential projection and
complementation: the former is easy whereas the latter is difficult.

For certain classes of automatic structures, it is possible to
avoid eliminating universal quantifiers via existential projection
and negation. For example, it was shown in~\cite{CH17} that deciding sentences of
quantified integer programming
$\exists{\bar{x}_1} \forall{\bar{x}_2}
\dots \exists{\bar{x}_n}. A \cdot \bar{x} \ge \bar{b}$ is complete for
the $n$th level of the polynomial hierarchy. The upper bound was
obtained by manipulating so-called hybrid linear sets, which
characterise the sets of integer solutions of systems of linear
equations $A \cdot \bar{x} \ge \bar{b}$. A key technique introduced
in~\cite{CH17} is called \emph{universal projection} and enables directly
eliminating universal quantifiers instead of resorting to double
complementation and existential projection. Given $S \subseteq
\N^{d+k}$, the universal projection of $S$ onto the first $d$
coordinates is defined as
\[
    \univProj[d]{S} \coloneqq \set[\big]{
        \bar{u} \in \N^d
        \suchthat (\bar{u},\bar{v}) \in S \text{ for all } \bar{v} \in \N^k }\,.
\]
It is shown in~\cite{CH17} that if $S$ is a hybrid linear set then
$\univProj[d]{S}$ is a hybrid linear set that can be obtained as a
finite intersection of existential projections of certain hybrid linear sets. Moreover, the growth of the
constants in the description of the hybrid linear set is only
polynomial. Neither syntactic quantifier elimination nor
automata-based methods are powerful enough to derive those tight upper
bounds for quantified integer programming.

Another example is a recent paper of Boigelot et al.~\cite{BFV23}
showing that, in an automata-theoretic approach for a fragment of
Presburger arithmetic with uninterpreted predicates, a universal
projection step can directly be carried out on the automata level
without complementation and only results in a singly exponential
blow-up.

Those positive algorithmic and structural results are specific to
Presburger arithmetic and leave open the option that it may be
possible to establish analogous results for general automatic
structures. The starting point of this paper is the question of
whether, given a non-deterministic finite automaton $\cA$ whose
language $\lang{\cA}\subseteq (\Sigma^{d+k})^*$ encodes the set of
solutions of some quantifier-free formula $\Phi$, there is a more
efficient way to eliminate a (block of) universally quantified
variable(s) than to first complement $\cA$, next to perform an
existential projection step, and finally to complement the resulting
automaton again, especially in the light of the results
of~\cite{CH17,CHM22}. Such a method would have direct consequences for
tools such as \textsc{Walnut} which perform the aforementioned
sequence of operations in order to eliminate universal quantifiers. In
particular, \textsc{Walnut} is not restricted to automata resulting
from formulas of linear arithmetic and allows users to directly
specify a finite-state automaton when desired.

For better or worse, however, as the main result of this paper, we
show that deciding whether the universal projection
$\univProj[d]{\lang{\cA}}$ of some regular language
$\lang{\cA} \subseteq \left(\Sigma^{d+k}\right)^*$ is empty is complete
for \expspace. In particular, the lower bound already holds for
$d=k=1$, meaning that, in general, even for fixed-variable fragments
of automatic structures, there is no algorithmically more efficient
way to eliminate a single universal quantifier than the na\"ive
one. The challenging part is to show the \expspace lower bound, which
requires an involved reduction from a tiling problem. This reduction
also enables us to show that there is a family $\big(\cA_n\big)_{n \in \N}$ of non-deterministic
finite automata such that $\size{\cA_n} = \BigO{n^4}$ and the smallest non-deterministic
finite automaton recognising the universal projection of
$\lang{\cA_n}$ has $\Omega\kern-1pt\left(2^{2^n}\right)$ states.

One of the most prominent automatic structures is B\"uchi arithmetic~\cite{B60,Bru85}. Given an integer $p\ge 2$, \emph{B\"uchi 
arithmetic of base $p$} is the first-order theory
of the structure $\langle \N,0,1,+, \pV_p{} \rangle$, where $\pV_p{}$ is a binary predicate such that $\pV_p{}(x,y)$ holds whenever $x$ is the largest power of $p$ dividing $y$ without remainder, i.e., $x=p^k$ for some $k\ge 0$, $x \mid y$ and $px \nmid y$. B\"uchi arithmetic
is a universal automatic structure in the following sense:
For any regular language $L \subseteq (\Sigma^n)^*$, there
is a formula \typeout{HEREA}$\Phi(\vec x)$\typeout{HEREB} of
B\"uchi arithmetic with an $\existsSymbol^*\forallSymbol^*$
quantifier prefix such that the set of satisfying
$\vec x \in \N^n$ is isomorphic to $L$~\cite{HR21}. The
existential fragment of B\"uchi arithmetic is NP-complete,\typeout{End}
and the full first-order theory is complete for \textsc{Tower}~\cite{GuepinH019} (see e.g.~\cite{Sch16} for a definition of \textsc{Tower}). While the computational complexity
of Presburger arithmetic with fixed quantifier alternation
prefixes is well 
understood~\cite{Haa14}, to the best of the authors' knowledge 
no (stronger) lower bounds are known when
generalizing to B\"uchi arithmetic.
A further contribution of this paper is to show that
B\"uchi arithmetic with an $\existsSymbol^*\forallSymbol^*\existsSymbol^*$ quantifier prefix is \expspace-hard, and it is \expexpspace-hard 
with a $\existsSymbol^*\forallSymbol^*\existsSymbol^*\forallSymbol^*$
quantifier prefix. Those lower bounds are obtained
by adapting the aforementioned $\expspace$ lower bound 
for universal projection emptiness.

    \section{Preliminaries}
        \subsection{Regular languages and their compositions}
        For a word $w = a_1 a_2 \cdots a_n \in \Sigma^*$, we write $w[i]$ to denote its $i$th letter $a_i$, and $w[i,j]$ to denote the infix $a_i a_{i+1} \cdots a_j$ ($i \le j$).
We write $\size w$ for the length of $w$.
A \emph{proper suffix} of $w$ is any infix $w[i,n]$ for some $1 < i \le n$.

\subparagraph*{Regular expressions}

A \emph{regular expression} over the alphabet $\Sigma$ is a term featuring Kleene star, concatenation and union operations, as well as $\emptyset$ and all symbols from $\Sigma$ as constants:
\begin{align*}
    \cE, \cE'
    &\coloncoloneqqq
        \cE^* \mid
        \cE \cdot \cE' \mid
        \cE + \cE' \mid
        \emptyset \mid
        a \text{ for every $a \in \Sigma$}
\end{align*}
For notational convenience, we also use sets of symbols $A \subseteq \Sigma$ as constants, and a $k$-fold concatenation $\cE^k$ for every $k \in \N$; we also drop the concatenation dot most of the time.
The language $\lang{\cE} \subseteq \Sigma^*$ is defined by structural induction, by defining
$\lang{\emptyset} \coloneqq \emptyset$ and
$\lang{a} \coloneqq \set{ a }$, and using the standard semantics of the three operations.
The class of languages definable by regular expressions is called \emph{regular languages}.
The size $\size \cE$ of a regular expression $\cE$ is defined recursively as $1$ plus the sizes of its subexpressions, where $\size a = \size \emptyset \coloneqq 1$.
For $\rho\colon \Sigma \to \Gamma$ and a regular expression $\cE$, $\rho(\cE)$ is a regular expression over $\Gamma$ obtained through substituting every constant $a \in \Sigma$ appearing in $\cE$ by $\rho(a)$.

\subparagraph*{Finite-state automata}
Regular languages can also be represented by \emph{non-deterministic
finite-state automata} (\NFA).  Such an automaton is a tuple $\cA
= \tuple{Q,\Sigma,\delta,Q_\ini,Q_\fin}$, where $Q$ is a finite
non-empty set of \emph{states}, $\Sigma$ is a finite \emph{alphabet},
$\delta \subseteq Q \times \Sigma \times Q$ is the \emph{transition
relation}, $Q_\ini \subseteq Q$ is the set of \emph{initial states},
and $Q_\fin \subseteq Q$ is the set of \emph{final states}.  A triple
$(p, a, q) \in Q \times \Sigma \times Q$ is called a \emph{transition}
and denoted as $p \xrightarrow{a} q$.  A \emph{run} of $\cA$ from a
state $q_0$ to a state $q_n$ ($n \in \N$) on a word $w=a_1 a_2 \cdots
a_n \in \Sigma^*$ is a finite sequence of transitions
$\bracket[\big]{q_{i-1} \xrightarrow{a_i} q_{i}}_{1 \le i \le n}$ such
that $q_{i-1} \xrightarrow{a_i} q_{i} \in \delta$ for every $i$.
A word $w \in \Sigma^*$ is \emph{accepted} by $\cA$ if
there exists a run of $\cA$ from some $q_\ini \in Q_\ini$ to $q_\fin \in Q_\fin$ over $w$.
The \emph{language} of $\cA$ is defined as $\lang{\cA} \coloneqq \set{ w\in \Sigma^* \suchthat w
\text{ is accepted by } \cA }$.
We define the size of $\cA$ as
%
    $\size \cA \coloneqq \size Q + \size \Sigma + \size \delta$.
%
Subsequently,
we will implicitly apply the well-known fact that the size 
of an \NFA accepting the complement of $\lang{\cA}$ is $\BigO{2^{\size{Q}} \cdot \size{\Sigma}}$,
and that it has $2^{\size{Q}}$ states.

Below we state, without proofs, a few folklore properties of \NFA:
\begin{fact}[\NFA closed under language union]
    For any \NFA $\cA, \cB$ over $\Gamma$ with
    states $Q_\cA$, $Q_\cB$, there exists an \NFA $\cA \oplus \cB$ with $\size{Q_\cA} + \size{Q_\cB}$ states such that $\lang{\cA \oplus \cB} = \lang{\cA} \cup \lang{\cB}$.
\end{fact}
Given finite alphabets $\Sigma, \Gamma$, a \emph{homomorphism}
is a function
$\rho\colon \Sigma^* \to \Gamma^*$ such that $\rho(v\cdot w)=\rho(v)\cdot \rho(w)$ 
for all $v,w\in \Sigma^*$.
It follows that $\rho$ is fully defined by specifying $\rho(a)$ for all $a\in \Sigma$.
\begin{exa}
    Let $\Sigma=\set{a,b}$, $\Gamma=\set{x,y}$, and define $\rho$ such that $\rho(a)=xy$ and $\rho(b)=\epsilon$. Then
    $\rho(abba)=xyxy$ and $\rho^{-1}(\set{xyxy})=\lang{b^*ab^*ab^*}$.
\end{exa}
\begin{fact}[\NFA closed under inverse homomorphisms]\label{fact:inv-hom}
    For any \NFA $\cA=(Q,\Sigma,\delta,Q_\ini,Q_\fin)$ and a homomorphism $\rho\colon \Sigma^* \to \Gamma^*$, there exists an \NFA $\rho^{-1}(\cA)$ with 
    $\size{Q}$ states such that $\lang{\rho^{-1}(\cA)} = \rho^{-1}(\lang{\cA})$.
\end{fact}

\begin{fact}[\NFA closed under concatenation of languages]
    For any \NFA $\cA, \cB$ with states $Q_\cA$ and $Q_\cB$, respectively, there exists an \NFA $\cA \odot \cB$ with $\size{Q_\cA} + \size{Q_\cB}$
    states such that $\lang{\cA \odot \cB} = \lang{\cA} \cdot \lang{\cB} \coloneqq \set{u \cdot v \suchthat u \in \lang{\cA} \text{ and } v \in \lang{\cB}}$.
\end{fact}

\begin{fact}[translating regular expressions to \NFA] \label{fact:regular-expression-to-nfa}
   There is a deterministic algorithm~\cite{regexToNFA} that, given
   a regular expression $\cE$,
   constructs an \NFA $\regexToAutomaton{\cE}$ 
   such that $\lang{\regexToAutomaton{\cE}} = \lang{\cE}$.
\end{fact}

\subparagraph*{Filters}
A \emph{filter} is an auxiliary term introduced to simplify the proofs in \cref{sec:emptiness-of-a-universal-projection-is-expspace-hard}, allowing for a modular design of regular languages.
Fix a finite alphabet $\Sigma$ and let $\F \coloneqq \set{\top, \bot}$.
Define homomorphisms $\inProj{}, \outProj{} \colon \bracket{\Sigma \times \F}^* \to \Sigma^*$ by their actions on a single letter
\begin{align*}
    &\inProj{a, b} \coloneqq a &
    &\outProj{a, \top} \coloneqq a \qquad \outProj{a, \bot} \coloneqq \varepsilon\,. \\
    &\text{\scriptsize(output every symbol from $\Sigma$)} & &\text{\scriptsize(output only symbols paired with $\top$)}
\end{align*}
A filter over an alphabet $\Sigma$ is any language $F \subseteq \bracket{\Sigma \times \F}^*$.
It induces a binary \emph{input-output relation} $\rel{F} \subseteq \Sigma^* \times \Sigma^*$ between input words $u$ and their subsequences $v$:
\begin{align*}
(u, v) \in \rel{F}
\ \ \mathrel{\overset{\text{\scriptsize def}}{\Longleftrightarrow}}\ \
u = \inProj{w} \text{ and } v = \outProj{w} \text{ for some $w \in F$}\,.
\end{align*}
We define $F(u) \coloneqq \set{ v \suchthat (u,v) \in \rel{F}}$ to be the set of all possible outputs of $F$ on $u$.

\subparagraph*{Filtering regular expressions}
A \emph{filtering regular expression} $\cF$ over alphabet $\Sigma$ is any regular expression over $\Sigma \times \F$.
We write ${\cF(w)} \coloneqq \lang{\cF}(w)$.
To simplify the notation, we only write the $\Sigma$ component of the constants, and underline parts of the expression.
A symbol $a$ appearing in an underlined fragment represents a pair $(a, \top)$, and in a fragment which is not underlined a pair $(a, \bot)$.
Intuitively, underlined portions correspond to parts of the words being output.
We apply the same notational convention to words $w \in (\Sigma \times \Phi)^*$.
Additionally, for $\rho\colon \Sigma \to \Gamma$, we abuse the notation and extend it to the naturally defined homomorphism of type $\Sigma \times \Phi \to \Gamma \times \Phi$, which just preserves the coordinate belonging to $\Phi$.
\begin{exa}
    Fix $A = \set{\w{a}, \w{b}, \w{c}, \dots, \w{z}}$.
    Consider a filtering regular expression $\cF$ and a word $w$, both over $A \cup \set{\wBlank}$:
    \begin{align*}
        \cF &\coloneqq \bracket{\wOut{A} \concat A^* \concat \wBlank}^* \wOut{A} \concat A^*&
        w &\coloneqq \w{nondeterministic\wBlank{}finite\wBlank{}automaton}\,.
    \end{align*}
    We have:
    \begin{align*}
        \cF(w) &= \set{\w{nfa}}\,, \\
        \cF &= \bracket[\big]{\bracket{A \times \set{\top}} \cdot \bracket{A \times \set{\bot}}^* \cdot \bracket{\wBlank, \bot}}^* \cdot \bracket{A \times \set{\top}} \cdot \bracket{A \times \set{\bot}}^*\,, \\
        \lang{\cF} &\mathrel{\makebox[\widthof{\(=\)}]{\(\ni\)}} \w{\wOut{\w{n}}ondeterministic\wBlank{}\wOut{\w{f}}inite\wBlank{}\wOut{\w{a}}utomaton}\,.
    \end{align*}
\end{exa}

        \subsection{Automatic relations}
        
Let $\Sigma$ be a finite alphabet such that $\wSharp \not\in \Sigma$. We
denote by $\Sigma_\wSharp \coloneqq \Sigma \cup \set{ \wSharp }$. Let $w_1,\ldots,w_k \in
\Sigma^*$ such that $w_i = a_{i,1}a_{i,2} \cdots a_{i,\ell_i}$, and
$\ell \coloneqq \max\{ \ell_1,\ldots,\ell_k \}$. For all $1\le i\le k$ and
$\ell_i<j \le \ell$, set $a_{i,j} \coloneqq \wSharp$. The \emph{convolution} $w_1
\otimes w_2 \otimes \cdots \otimes w_k$ of $w_1,\ldots,w_k$ is defined
as
\[
    w_1 \otimes w_2 \otimes \cdots \otimes w_k \coloneqq
    { \begin{bmatrix}
        a_{1,1} \\
        \vdots  \\
        a_{k,1}
    \end{bmatrix}\begin{bmatrix}
                     a_{1,2} \\
                     \vdots  \\
                     a_{k,2}
    \end{bmatrix}\cdots\begin{bmatrix}
                           a_{1,\ell} \\
                           \vdots     \\
                           a_{k,\ell}
    \end{bmatrix}}\subseteq \bracket[\Big]{\Sigma_\wSharp^k}^*\,.
\]
For $R \subseteq \bracket[\big]{\Sigma^*}^k$ and $L \subseteq \bracket[\big]{\Sigma_\wSharp^{k}}^*$ define
\begin{align*}
    \relToLang{R} &\coloneqq \set{ w_1 \otimes w_2 \otimes \cdots \otimes w_k \suchthat (w_1, w_2,\ldots,w_k) \in R }\,, \\
    \langToRel{L} &\coloneqq \set{ (w_1, w_2, \dots, w_{k}) \suchthat w_1 \otimes w_2 \otimes \cdots \otimes w_k \in L}\,.
\end{align*}
In this paper, we say that a relation $R\subseteq (\Sigma^*)^k$ is \emph{automatic} whenever
$\relToLang{R}$ is regular.
%
Furthermore, we
assume that $\relToLang{R}$ is given by some \NFA
$\cA_R=(Q,\Sigma_\wSharp^k,\delta,Q_\ini,Q_\fin)$.

Clearly, not every \NFA $\cA=(Q,\Sigma_\wSharp^k,\delta,Q_\ini,Q_\fin)$ is associated
with an automatic relation $R\subseteq \Sigma^k$ since there are
\emph{a priori} no restrictions on the occurrences of the padding
symbol ``$\wSharp$''. The language $L_\inc \subseteq
(\Sigma_\wSharp^k)^*$ of all incorrect words that cannot be obtained as a
convolution of words $w_1,\ldots,w_k\in \Sigma^*$ can be characterized
by the following regular expression:
\[
    \bracket[\big]{\Sigma^k_\wSharp}^*
    \cdot
    \squareBracket[\Big]{
    \{ \wSharp \}^k \languagesum
    \textstyle\biglanguagesum_{1\le i \le k}
        \bracket[\Big]{
            \bracket[\big]{
                \Sigma_\wSharp^{i-1}
                \times
                \set{ \wSharp }
                \times
                \Sigma_\wSharp^{k-i}
            }
            \cdot
            \bracket[\big]{
                \Sigma_\wSharp^{i-1}
                \times
                \Sigma
                \times
                \Sigma_{\wSharp}^{k-i}
            }
        }
    }
    \cdot \bracket[\big]{\Sigma_\wSharp^k}^*\,.
\]
%
This regular expression ``guesses'' that either a letter consisting
solely of $k$ $\wSharp$ symbols occurs, or in some row of a word in
$\bracket{\Sigma_\wSharp^k}^*$ a ``$\wSharp$'' symbol is followed by a
symbol in $\Sigma$. The language of this regular expression can be
implemented by an \NFA with $k+2$ states. Hence, the complement
$L_\corrrr \coloneqq \setComplement{L_\inc}$ of $L_\inc$, characterizing
all ``good'' words, can be recognized by an \NFA with $2^{k+2}$
states. For the sake of readability, we do not parametrise $L_\inc$
explicitly with $k$; the relevant $k$ will always be clear from the
context.

The \emph{existential projection} of $R\subseteq (\Sigma^*)^{d+k}$
onto the first $d$ components is defined as
%
\begin{align*}
    \proj[d]{R} &\coloneqq
    \set[\big]{
        \bar{u} \in (\Sigma^*)^d
        \suchthat
        (\bar{u},\bar{w}) \in R \text{ for some } \bar{w} \in (\Sigma^*)^k
    }\,.
\end{align*}
The dual of existential projection is \emph{universal projection}:
\begin{align*}
    \univProj[d]{R} &\coloneqq
    \set[\big]{
        \bar{u} \in (\Sigma^*)^d
        \suchthat
        (\bar{u},\bar{w}) \in R \text{ for all } \bar{w} \in (\Sigma^*)^k
    }\,.
\end{align*}
%
It is clear that $\univProj[d]{R} =
\setComplement{\proj[d]{\setComplement{R}}}$.
We overload the projection notation for languages
\begin{align*}
    \proj[d]{L} &\coloneqq \relToLang[][][\big]{\proj[d]{\langToRel{L}}} &
    \univProj[d]{L} &\coloneqq \relToLang[][][\big]{\univProj[d]{\langToRel{L}}}\,.
\end{align*}
In this article, given $\cA_R$ such
that $\relToLang{R} = \lang{\cA_R} \subseteq \bracket[\big]{\Sigma_\wSharp^{d+k}}^*$, we are concerned with the
computational complexity of deciding whether
$\univProj[d]{R} = \emptyset$, measured in terms of $\size{\cA_R}$.
In \cref{sec:emptiness-of-a-universal-projection-is-expspace-hard,sec:upper-bounds} we will prove the following.
\begin{thm} \label{thm:main-theorem}
    Deciding whether $\univProj[d]{R}\neq \emptyset$ for an automatic
    relation $R\subseteq (\Sigma^*)^{d+k}$ 
    given by an \NFA 
    $\cA_R$ is \expspace-complete. The lower bound already holds for
    $d=k=1$.
\end{thm}

        \subsection{\Buchi arithmetic}
        
As discussed in the introduction, given an integer $p\ge2$, \emph{B\"uchi arithmetic of base $p$} is the first-order theory
of the structure $\langle \N,0,1,+, \pV_p{} \rangle$, where $\pV_p{}$ is a binary predicate such that $\pV_p{}(x,y)$ holds whenever $x$ is the largest power of $p$ dividing $y$ without remainder, i.e., $x=p^k$ for some $k\ge 0$, $x \mid y$ and $px \nmid y$. Atomic formulas of \Buchi arithmetic are $V_p$ predicates
or linear inequalities of the form $a_1 \cdot x_1 + \cdots + a_n \cdot x_n \ge b$, where all numbers are encoded in binary.
We use some standard abbreviations such as strict inequalities, etc.
Note that $\pV_p{x,x}$ holds if and only if
$x$ is a power of $p$, and we write $\pP_p{x}$ to abbreviate
$\pV_p{x,x}$. Given a formula $\Phi(x_1,\ldots,x_n)$
of B\"uchi arithmetic, we define $\sem{\Phi} \coloneqq \{ (m_1,\ldots,m_n) \in \N^n: \Phi(m_1,\ldots,m_n) \text{ is valid} \}$ to be its set of satisfying assignments. Given a string $w=d_0\cdots d_n\in \set{0,\ldots, p-1}^+$, we define $\sem{w}_p\in \N \coloneqq \sum_{i=0}^{n} p^i \cdot d_i$.

    \section{Emptiness after universal projection is \expspace-hard}
    \label{sec:emptiness-of-a-universal-projection-is-expspace-hard}

        \subsection{Tiling problems}
        Let $\cT \subseteq_\finite \N^4$ be a set of \emph{tiles} with colours coded as four-tuples of numbers with
associated projections $\tTop{}, \tRight{}, \tBottom{}, \tLeft{}\colon \N^4 \to \N$ to access individual colours of a tile, and let $\colours{\cT} \coloneqq \tTop{\cT} \cup \tRight{\cT} \cup \tBottom{\cT} \cup \tLeft{\cT}$.
\begin{exa}
A tile $t = (2, 4, 3, 3)$ is depicted as
\smash{$
    \begin{tikzpicture}[baseline=-0.5ex]
        \tileInner{2}{4}{3}{3}
    \end{tikzpicture}
$} with various auxiliary background shades corresponding to colour values.
\end{exa}
A $\cT$-\emph{tiling of size $(h, w) \in \N_+^2$} is an $h \times w$ matrix $T = \myMatrix{t_{i,j}}_{i, j} \in \cT^{h\times w}$.
It is \emph{valid} whenever colours of the neighbouring tiles match:
\begin{align}
    \tBottom{t_{i, j}} &= \tTop{t_{i+1, j}}   \label{eq:tiling-property-a}&  &\text{for every $1 \le i \le h - 1$ and $1 \le j \le w$,} \\
    \tRight{t_{i, j}}  &= \tLeft{t_{i, j+1}}  \label{eq:tiling-property-b}&  &\text{for every \rlap{$1 \le i \le h$}\phantom{$1 \le i \le h - 1$} and $1 \le j \le w - 1$.}
\end{align}
See \Cref{fig:problem-instance-and-solution} on \cpageref{fig:problem-instance-and-solution} for an example of a valid tiling. A \emph{$\cT$-tiling of width $w \in \N_+$} is any tiling in $\cT^{h\times w}$ for some $h \in \N_+$.
We define
\begin{align*}
    \cT^{\star \times w} \coloneqq \bigcup_{h \in \N_+} \cT^{h\times w}\,.
\end{align*}
Additionally, for two distinguished tiles $\toplefttile, \bottomrighttile \in \cT$, let $(\cT, \toplefttile, \bottomrighttile)$-tiling be any $\cT$-tiling with $\toplefttile$ placed in its top-left corner, and $\bottomrighttile$ in its bottom-right corner.

\begin{prob} \label{problem:tiling-problem} $\tilingProblem$
    \newcommand{\myWidth}{21mm}
    \ \par\smallskip\noindent
    \begin{tcolorbox}[blanker,sidebyside,sidebyside gap=2.7mm,sidebyside align=top seam,parbox=true,lefthand width=\myWidth,boxsep=0pt,top=0pt,bottom=0pt,before skip=0pt,after skip=0pt,oversize, halign lower=left,halign=right]
        \strut\normalfont\textsc{{Input:}}
        \tcblower
        \strut A $4$-tuple $(\cT,\toplefttile, \bottomrighttile, n)$, where
        \begin{itemize}
            \item $\cT \subseteq_{\finite} \N^4$ is a finite set of tiles,
            \item $\toplefttile, \bottomrighttile \in \cT$,
            \item $n \in \N_+$ given in unary.
        \end{itemize}
    \end{tcolorbox}

    \medskip\noindent
    \begin{tcolorbox}[blanker,sidebyside,sidebyside gap=2.7mm,sidebyside align=top seam,parbox=true,lefthand width=\myWidth,boxsep=0pt,top=0pt,bottom=0pt,before skip=0pt,after skip=0pt,oversize, halign lower=justify,halign=right]
        \strut\normalfont\textsc{{Question:}}
        \tcblower
        \strut Does there exist a valid $(\cT, \toplefttile, \bottomrighttile)$-tiling of width $2^n$?
    \end{tcolorbox}
\end{prob}
By $\mathbb{T} \subset \mathcal{P}_\finite(\N^4) \times \N^4 \times \N^4 \times \N_+$ we denote the set of all 
syntactically valid instances of the above problem.

\smallskip
\begin{lem}
    $\tilingProblem$ 
    is \expspace-hard.
\end{lem}
\begin{proof}[Proof (Sketch)]
It is part of the folklore of the theory of computation that tiling problems can simulate the computation of Turing machines,
the width of the requested tiling corresponding to the length of tape the machine is allowed to use.
\expspace-completeness of a variant similar to the one above is sketched in~\cite{TilingProblems}, see also the \textsc{Corridor Tiling} problem in~\cite{vanEmde19}
and the discussion of it capturing computations of space-bounded Turing machines.

In more technical terms, let $M$ be a Turing machine using at most $2^{p(m)}$ tape cells
on an input $w \in \{0,1\}^*$ of length $m$ for some polynomial $p\colon \N \to \N$. 
To decide whether $M$ accepts $w$, we show how to construct in logarithmic space an
instance $(\cT,\toplefttile, \bottomrighttile, n)$ of $\tilingProblem$ from 
an encoding of $M$ and $w$ such that a valid tiling exists if and only if $M$ has an accepting
computation on input $w$. Note that any computation of $M$ can be represented by a computation
table of width $2^{p(m)}$ such that the $i$th row of that table represents the full configuration
of $M$ after $i$ computation steps. Each cell of the computation table is either a symbol
of the working tape, or such a symbol paired with a control state $q$ of $M$ to indicate that 
$M$ is currently reading that symbol in state $q$. Consistency of a computation table can be verified locally: in every row, only one cell paired with some $q$ is allowed to occur, and 
between rows, the movement of the head, the state transition and the updating of cells occur 
only locally around the head position. 
This local consistency property is precisely what allows the computation table
to be enforced by local tiling constraints in a $\tilingProblem$ instance.
Finally, requiring the initial state of $M$ to appear
in the first row and the accepting state in the last row of the computation table ensures
that the table represents an accepting computation of $M$.

To define $(\cT,\toplefttile, \bottomrighttile, n)$, we set $n=p(m)+1$ and choose a 
finite set of tiles that allow us to represent a computation table as described above. The 
horizontal adjacency constraints enable us to adhere to the requirement that a 
state appears only once along every row, and vertical constraints can be used to ensure
correct state, head, and tape updates between consecutive computation steps. 
The special tile $\toplefttile$ is chosen such that the first row of a tiling 
represents a valid inital configuration of $M$ reading $w$, and whenever
an accepting state occurs in some row, the tiles are chosen such that $\bottomrighttile$ 
can appear at the very last tile of the tiling.
\end{proof}

        \subsection{The reduction}
        We prove \cref{thm:main-theorem} by a reduction from $\tilingProblem$ and show that \expspace-hardness occurs in the simplest case of universal projection: projecting a binary relation to get a unary one.
Intuitively, for each instance $\I = (\cT, \toplefttile, \bottomrighttile, n)$ of $\tilingProblem$, we construct an automaton $\cA_\I$ such that $\univProj[1]{{\lang{\cA_\I}}}$ is not empty if and only if $\I$ is a \textsc{Yes}-instance.
Formally, we provide a family of \logspace-constructible \NFA $(\cA_\I)_{\I \in \mathbb{T}}$ over the alphabet $(\Sigma_\I \cup \set{\wSharp})^2$ for some $\Sigma_\I$, each of size $\BigO{n^4 \cdot \size{\cT}^3}$, representing  the relation $\automatonToRelation{\cA_\I} \subseteq (\Sigma_\I^*)^2$, such that
\begin{align} \label{eq:reduction-property}
    \univProj[1]{\lang{\cA_\I}} \neq \emptyset \iff \text{there exists a valid $(\cT, \toplefttile, \bottomrighttile)$-tiling of width $2^n$}.
\end{align}
The \logspace constructibility is easy to verify and we chose not to provide explicit arguments for it.
For the remainder of this section, we fix an instance $\I \in \mathbb{T}$.
For technical reasons, we assume with no
loss of generality that $n \geq 6$.

In \cref{subsec:word-encoding-of-tiling},
we define $\Sigma_\I$, specify a language $\encLang \in \Sigma_\I^*$, and prove that:
\begin{restatable}{lem}{correctEncodingsDesignChoice}
\label{lem:correct-encodings-design-choice}
    $\encLang \neq \emptyset \iff \text{there exists a valid $(\cT, \toplefttile, \bottomrighttile)$-tiling of width $2^n$}$.
\end{restatable}
In turn in \cref{subsec:construction-of-the-automaton}, we construct an \NFA $\cA_\I$ such that
\begin{lem} \label{lem:projection-of-relation-equals-valid-encodings}
    $\univProj[1]{\lang{\cA_\I}} = \encLang$.
\end{lem}
This completes the proof of \cref{thm:main-theorem}, the correctness of the reduction stemming directly from \cref{lem:correct-encodings-design-choice,lem:projection-of-relation-equals-valid-encodings}.


%

        \subsection{Word encoding of tilings} \label{subsec:word-encoding-of-tiling}
        Here, we provide $\Sigma_\I$ and an encoding $\enc{}\colon \cT^{\star \times 2^n} \to \Sigma_\I^*$.
Then we define $L_\I$ as an intersection of six languages, and prove \cref{lem:correct-encodings-design-choice} by showing that it coincides with the language of encodings of valid tilings.

Let $N_n \coloneqq \N \cap [0, n]$.
Additionally, let $N_n^{{\sim} k} \coloneqq \set{ i \in N_n \suchthat i \mathop{{\sim}} k }$ for ${{\sim}} \in \set{<,=,>}$ and $k \in \N$.
The alphabet $\Sigma_\I$ consists of three groups of symbols -- tiles from $\cT$, numbers from $N_n$, and auxiliary symbols:
\begin{align*}
    \Sigma_{\I} \coloneqq \cT \cup N_n \cup {\set{\wA, \wRB, \wRE, \wCB, \wCE}}\,.
\end{align*}
The symbol $\wA$ is a mnemonic marking places where we enforce ``for-all''-type properties.
In what follows, we colour some symbols (\eg, $\w{\wPref{3010}}\concat{t}\concat{\w{\wSuf{20103}}}$) to assist in understanding the construction; such designations are auxiliary and are not reflected in the alphabet.
The encoding of runs uses the word $\comb{n} \in N_n^*$
\begin{align*}
    \comb{n} &\coloneqq {n} \concat \combInner{n-1} \concat {n}\,,
        \intertext{where the words \smash{$\bracket{\combInner{i}}_{0 \le i \le n}$} are defined recursively as}
    \combInner{0} &\coloneqq \w0 \\
    \combInner{i} &\coloneqq \combInner{i-1} \concat {i} \concat \combInner{i-1} & &\text{for $0 < i \le n$.}
\end{align*}
Observe that $\comb{n}$ has length exactly $2^n + 1$. 
The recursive definition of $\comb n$ will later enable us to recognise 
$\{\comb{n}\}$
as an intersection of $n$ \NFA, each with a constant number of states.
The intersection itself will be implemented as a universal projection step.

\begin{exa}
    $\comb{4}$ is $\w{\wWordIV}$ and has length $17$.
\end{exa}

The following observation shows that $\comb{n}$ is tightly related to a binary counter on $2^n$ bits, which will play a crucial role in \cref{sec:second-buchi-lower-bound}.

\begin{obs} \label{obs:crucial-observation} Consider a $2^n$-bit counter counting from $0$ to $2^{2^n}-1$ and overflowing back to $0$:
\begin{align*}
    \underbrace{0{\ldots}00}_{\text{$2^n$ times}} \xrightarrow{0} 0{\ldots}01 \xrightarrow{1} 0{\ldots}10 \xrightarrow{0} 0{\ldots}11 \xrightarrow{2} \cdots \to 1{\ldots}10 \xrightarrow{1} \underbrace{1{\ldots}11}_{\text{$2^n$ times}}\xrightarrow{N} \underbrace{0{\ldots}00}_{\text{$2^n$ times}}
\end{align*}
Let $w_i$ be the binary representation of $i$ on $2^n$ bits; set $w_{2^{2^n}} \coloneqq w_0$.
We index bits of $w_i$ starting from the least significant one (index $0$) to the most significant one (index $2^n-1$).
The $i$th letter of $\comb{2^n}$ corresponds to the smallest index at which $w_i$ and $w_{i+1}$ differ for $i \in \set{0,\ldots, 2^{2^n}-1}$.
Furthermore, for every such $i$, the $j$th bit of $w_i$ is $1$ if and only if there is a position $i'$ to the left of $i$ such that $\comb{2^n}[i'] = j$ and the infix $\comb{2^n}[i'+1, i]$ contains only symbols smaller than $j$.

An example involving $\comb{4}=40102\dots04$ is given in \cref{fig:crucial-observation}. We index the positions between the letters of $\comb{4}$, starting from 0 and displayed in a smaller grey font. Those positions  correspond to the binary strings $w_i$ above. We highlight positions 5 and 11 and illustrate how to obtain the binary 
representations $w_5$ and $w_{11}$ based on the properties of the comb. To this end, above every element $k$ of the comb, we vertically (top to bottom) write the binary string $1\cdot0^k$.
In order to obtain, e.g, bit $b_1$ of the bit string 
$w_5=b_3 b_2 b_1 b_0$,  we search to the left for the closest
symbol ``visible'' at level 1, which is a 0 from the bit string
$100$, and which determines the value of $b_1\coloneqq 0$.
Likewise, the most significant bit of
$w_{11}$ is 1 since at position 8, $\comb{4}$ takes value 3 and
is hence labelled with $1000$.
\end{obs}

\begin{figure}
\begin{tikzpicture}[>=stealth,xscale=0.79, yscale=0.4]
  \draw[->] (-4.5,-0.65) -- (12.5,-0.65);
  \foreach \x/\h in { 
     -4/4,-3/0, -2/1, -1/0, 0/2, 
      1/0,   2/1,   3/0,   4/3,   5/0, 
      6/1,   7/0,   8/2,   9/0,  10/1, 
     11/0,  12/4                        
  }{
    \node at (\x,-1.3) {\small\h};
    \pgfmathtruncatemacro{\xpp}{\x + 4}
    \ifnum\xpp<16
    \node at (\x+0.5,-2) {\scriptsize\textcolor{gray}{\xpp}};
    \fi
    \ifnum\h>0
    \foreach \y in {1,...,\h} {
        \node at (\x,\y-1) {\scriptsize{0}};
    }
    \fi
    \node at (\x,\h) {\scriptsize{1}};
    \draw (\x-0.2,0-0.4) rectangle (\x+0.2,\h+0.4);
  }

  \draw[line width=2pt,cA] (1.7,-2) -- (1.7,3.4);
  \draw[line width=2pt,cD] (7.7,-2) -- (7.7,3.4);

  \node[inner sep=0] (p1) at (1.5,0) {\footnotesize{$\mathbf{1}$}};
  \node[inner sep=0] (p2) at (1.5,1) {\footnotesize{$\mathbf{0}$}};
  \node[inner sep=0] (p3) at (1.5,2) {\footnotesize{$\mathbf{1}$}};
  \node[inner sep=0] (p4) at (1.5,3) {\footnotesize{$\mathbf{0}$}};
  
  \node[inner sep=0] (p1) at (7.5,0) {\footnotesize{$\mathbf{1}$}};
  \node[inner sep=0] (p2) at (7.5,1) {\footnotesize{$\mathbf{1}$}};
  \node[inner sep=0] (p3) at (7.5,2) {\footnotesize{$\mathbf{0}$}};
  \node[inner sep=0] (p4) at (7.5,3) {\footnotesize{$\mathbf{1}$}};
  
  \begin{pgfonlayer}{background}
    \foreach \x/\y in {1.3/0,0.3/1,0.3/2,-3.7/3} {
        \fill[cA,opacity=0.15] (\x-0.5,\y-0.4) rectangle (1.7,\y+0.4);
    }
    \foreach \x/\y in {7.3/0,6.3/1,4.3/2,4.3/3} {
        \fill[cD,opacity=0.15] (\x-0.5,\y-0.4) rectangle (7.7,\y+0.4);
    }
  \end{pgfonlayer}

  \node[fill=cA,text=white,inner xsep=3.1pt,inner ysep=1.5pt] at (1.5,-2) {\scriptsize\textbf5};
  \node[fill=cD,text=white,inner xsep=1.2pt,inner ysep=1.5pt] at (7.5,-2) {\scriptsize\textbf{1\kern-1pt1}};
  \node at (-5.4,-1.3) {\small$\comb{4} = {}$};
  \node at (-5.15,0) {\scriptsize $b_0$};
  \node at (-5.15,1) {\scriptsize $b_1$};
  \node at (-5.15,2) {\scriptsize $b_2$};
  \node at (-5.15,3) {\scriptsize $b_3$};

   \foreach \x in {0,...,15} {
    \node at (\x-3.5,-3) {\scriptsize\color{gray}$w_{\x}$};
  }
\end{tikzpicture}
\caption{Illustration of $\comb{4}$ and how to decode positions in a comb.}
\label{fig:crucial-observation}
\end{figure}

%
%
We define the \emph{encoding} function $\enc{}\colon \cT^{\star \times 2^n} \to \Sigma_\I^*$ in three steps.
Let $T = \myMatrix{t_{i,j}}_{i,j} \in \cT^{h \times 2^n}$ for some $h \in \N$.
The tile $t_{i,j}$ in $T$ is represented as
\begin{align*}
    \encCell{T,i,j} &\coloneqq \wCB \concat \text{\wPref*{$\comb{n}[1, j]$}} \concat {t_{i,j}} \concat \text{\wSuf*{$\comb{n}[\kern1pt j+1, 2^n+1]$}} \concat \wA \concat \wCE\,,
    \intertext{a single row is encoded as}
    \encRow{T,i} &\coloneqq \wRB \concat {\textstyle\biglanguageconcat_{1 \leq j \leq 2^n} \encCell{T,i,j}} \concat \wRE\,,
    \intertext{and finally, the encoding of the entire tiling is defined as}
    \enc{T} &\coloneqq \wEncodingPrefix \concat \textstyle\biglanguageconcat_{1 \leq i \leq h} \encRow{T,i}\,.
\end{align*}
\phantom{a}\vspace{-\baselineskip}
\begin{exa}
    The tiling $T = \myMatrix{t_{i,j}}_{i,j}$ of size $(2, 2^4)$ is encoded as
    \begingroup\small
    \begin{align*}
        \wEncodingPrefix \concat &\wRB
        \wCB \w{\wPref{4}}\,t_{1,1}\,\w{\wSuf{0102010301020104}} \concat \wA \wCE \cdots
        \wCB \w{\wPref{40102}}\,t_{1,5}\,\w{\wSuf{01030$\cdots$04}} \concat \wA \wCE \cdots
        \wCB \w{\wPref{4010201030102010}}\,t_{1,16}\,\w{\wSuf{4}} \concat \wA \wCE \wRE \cdot {} \\
        {} \cdot {} &\wRB
        \wCB \w{\wPref{4}}\,t_{2,1}\,\w{\wSuf{0102010301020104}} \concat \wA \wCE \cdots
        \wCB \w{\wPref{40102}}\,t_{2,5}\,\w{\wSuf{01030$\cdots$04}} \concat \wA \wCE \cdots
        \wCB \w{\wPref{4010201030102010}}\,t_{2,16}\,\w{\wSuf{4}} \concat \wA \wCE \wRE.
    \end{align*}
    \endgroup
    The word above is written in two lines to make the correspondence to a tiling more apparent.
\end{exa}

\paragraph*{Languages of encodings}

Define the language of encodings of valid tilings of width $2^n$ with $\toplefttile, \bottomrighttile$ in the correct corners
\begin{align*}
    \valEncLang \coloneqq \set[\big]{\enc{T} \suchthat \text{$T$ is a valid $(\cT, \toplefttile, \bottomrighttile)$-tiling of width $2^n$}}\,.
\end{align*}
In order to express the notion of an encoding of a valid tiling in a more tangible way,
below we define languages $\condLang1, \dots, \condLang6$,
which, as we prove in \cref{lem:correct-encodings-design-choice-more-precise}, jointly characterise encodings.
The first three are given by regular expressions fixing
basic properties of the encoding, the next two guarantee an appropriate width of the encoding,
while the last one enforces in a non-trivial way that the colours match
vertically.

\begin{cond} 
Recall that the alphabet $\Sigma_\I$ is
parametrised by $n$.
Language $\condLang1$ is given by the regular expression
    \label{condition:first}
    \label{condition:rows-of-cells-structure}
    \begin{align*}
        \cE_\I^1 \coloneqq
            \left(\wRB \concat \wCB \concat \wPref{n} \concat \cT \concat \wSuf{N_n^*} \concat \wA \wCE \concat \Big(\wCB \concat \wPref{N_n^*} \concat \cT \concat \wSuf{N_n^*} \concat \wA \concat \wCE\Big)^* \concat \wCB \concat \wPref{N_n^*} \concat \cT \concat \wSuf{n} \concat \wA \wCE \wRE\right)^*\,.
    \end{align*}
    Intuitively, encodings consist of rows bounded by $\wRB$ and $\wRE$; each row comprises cells delimited by $\wCB$ and $\wCE$. The first cell begins with the number $n$ followed by a tile, while the last one ends with a tile, $n$ and $\wA$. 
    Observe that while the size of $\cB_\I^1 \coloneqq\regexToAutomaton{\cE^1_\I}$ grows
    linearly in $n$, its number of states is constant.
\end{cond}

\begin{cond}
\label{condition:corners-fixed}
    The language $\condLang2$ is defined by the regular expression
    \begin{align*}
        \cE_\I^2 \coloneqq
            \wRB \wCB \wPref{N_n^*} \concat \toplefttile \concat \wSuf{N_n^*} \concat \wA \wCE \concat \Big(\wCB \wPref{N_n^*} \concat \cT \concat \wSuf{N_n^*} \concat \wA \wCE\Big)^*
         \concat \wRE \concat
            \concat \Sigma_\I^* \concat \concat
            \wRB \concat \Big(\wCB \wPref{N_n^*} \concat \cT \concat \wSuf{N_n^*} \concat \wA \wCE\Big)^* \concat \wCB \wPref{N_n^*} \concat \bottomrighttile \concat \wSuf{N_n^*} \concat \wA \wCE \wRE\,.
    \end{align*}
    This requires the first row of a purported tiling to begin with $\toplefttile$, and the last row to end with $\bottomrighttile$.
    As in \cref{condition:rows-of-cells-structure}, $\condLang2$ is obtained by an \NFA $\cB_\I^2 \coloneqq \regexToAutomaton{\cE_\I^2}$ with $\BigO{1}$ states.
\end{cond}

\begin{cond}
    \label{cond:colours-match-horizontally}
    Let $Q \coloneqq \colours{\cT}$ and $\cB^3_\I \coloneqq (Q, \Sigma_\I, \delta, Q, Q)$, where $\delta$ has transitions
    \begin{align*}
        i &\xrightarrow{t} j & &\text{for every $i,j \in Q$ and $t \in \cT$ such that $\tLeft{t} = i$ and $\tRight{t} = j$,} \\
        i &\xrightarrow{a} i & &\text{for every $i \in Q$ and $a \in \Sigma_\I \setminus {(\cT \cup \set{\wRE})}$,} \\
        i &\xrightarrow{\wRE} j & &\text{for every $i, j \in Q$.}
    \end{align*}
    We set $\condLang3 \coloneqq \lang{\cB^3_\I}$; it contains encodings where tile colours match horizontally. Observe that $\size Q=\BigO{\size \cT}$.
\end{cond}

\begin{cond}[each cell contains a $\comb{n}$]
    \label{cond:cells-contain-combs} The definition of $\condLang4$ uses a filtering regular expression $\cF^4_\I$:
    \begin{align*}
        \cF^4_\I &\coloneqq
        \wCB \concat \wPref{\wOutD{N_n^*}} \concat \cT \concat \wOutD{\wSuf{N_n^*} \concat \wA} \concat \wCE \concat \Sigma_\I^* \\
        \condLang4 &\coloneqq \set{w \in \Sigma_\I^* \suchthat \comb{n} \concat \wA \in \cF^4_\I(v) \text{ for every proper suffix $v$ of $w$ such that $v[1] = \wCB$}}
    \end{align*}
\end{cond}

\begin{cond}
    [prefix of a cell and first symbols of following cells' suffixes form a $\comb{n}$]
    \label{cond:cells-prefix-and-first-letters-of-suffixes-form-a-comb}
    \begin{align*}
        \cF^5_\I &\coloneqq
        \wCB \concat \wRedC{\wOutD{N_n^* }} \concat \cT \concat \wBlueC{\wOutD{N_n} \concat N_n^*} \concat \wA \concat \wCE
        \bracket{\wCB \concat \wRedC{N_n^* } \concat \cT \concat \wBlueC{\wOutD{N_n} \concat N_n^*} \concat \wA \concat \wCE}^*
        \concat \wCB \concat \wRedC{N_n^* } \concat \cT \concat \wBlueC{\wOutD{N_n} \concat N_n^*} \concat \wOutD{\wA} \concat \wCE
        \concat \wRE \concat \Sigma_\I^* \\
        \condLang{5} &\coloneqq \set{w \in \Sigma_\I^* \suchthat \comb{n} \concat \wA \in \cF^5_\I(v) \text{ for every proper suffix $v$ of $w$ such that $v[1] = \wCB$}}
    \end{align*}
\end{cond}

\begin{cond}[tile colours match vertically]
    \label{condition:last}
    \label{cond:colours-match-vertically}
    Let ${\blacktriangledown_t} \coloneqq \set{t' \in \cT \suchthat \tTop{t'} = \tBottom{t}}$ be the set of tiles whose top colour matches the bottom colour of tile $t$.
    Define
    \newcommand{\hlp}[1]{\makebox[4mm][c]{#1}}
    \begin{align*}
        \cF^6_\I \coloneqq
        \textstyle\biglanguagesum_{t \in \cT}
        \Big(\mspace{147mu}\wCB \concat \wPref*{\wOutD{N_n^*}} \concat &\hlp{$t$} \wSuf*{\concat N_n^*} \concat  \wA \concat \wCE
        \concat \bracket[\big]{\wCB \concat \wPref*{N_n^*}  \concat \cT \concat  \wSuf*{N_n^*} \concat \wA \concat \wCE}^* \concat \wRE
        \cdot {} \\
        \wRB \bracket[\big]{\wCB \concat \wPref*{N_n^*}  \concat \cT \concat  \wSuf*{N_n^*} \concat  \wA \concat \wCE}^* \concat
        \wCB \concat \wPref*{N_n^*} \concat &\hlp{$\blacktriangledown_t$} \concat \wOutD{\wSuf*{N_n^*} \concat \wA} \concat \wCE \concat
        \bracket[\big]{\wCB \concat \wPref*{N_n^*}  \concat \cT \concat  \wSuf*{N_n^*} \concat  \wA \concat \wCE}^* \concat  \wRE \concat \Sigma_\I^* \Big)\,.
    \end{align*}
    The expression above was typeset in two lines only to highlight the correspondence between cells in two consecutive rows.
    Define the language $\condLang{6}$ as
    \begin{align*}
        \condLang{6} \coloneqq {} &\big\{w \in \Sigma_\I^* \:\big\vert\: \comb{n} \concat \wA \in \cF^6_\I(v) \text{ for every proper suffix $v$ of $w$ such that} \\
        &\text{\phantom{$\big\{w \in \Sigma_\I \:\big\vert\:$}$v[1] = \wCB$ and $v[j] = \wRB$ for some $j\hspace{44mm}$}\big\}
    \end{align*}
    Intuitively, requiring $\wRB$ to appear in $v$ ensures that
    $v$ does not begin in the last row.    
\end{cond}
Observe that $\regexToAutomaton{\cF^4_\I}$ and $\regexToAutomaton{\cF^5_\I}$ have a constant number of states, and $\regexToAutomaton{\cF^6_\I}$ has $\BigO{\size{\cT}}$ states. Define $L_\I \coloneqq \wEncodingPrefix \concat \bigcap_{1 \leq
i \leq 6} \condLang{\kern1pt i}$. Let us recall:
\correctEncodingsDesignChoice*
This lemma directly follows from:
\begin{lem}
    \label{lem:correct-encodings-design-choice-more-precise}
    $L_\I = \valEncLang$.
\end{lem}
\begin{proof}
    The inclusion $L_\I \supseteq \valEncLang$ is trivial.
    \proofCase{Inclusion $L_\I \subseteq \valEncLang$}
    Take any $u \in L_\I$.
    Due to \cref{condition:rows-of-cells-structure}, it has the form
        $\wEncodingPrefix \concat \biglanguageconcat_{1 \le i \le h} \wRB \concat v_i \concat \wRE$, 
    for some $h \in \N$, where each $v_i \in (\Sigma_\I \setminus \set{\wRB, \wRE})^*$.
    We will show that $\wRB \concat v_i \concat \wRE$ is a proper 
    encoding of a row under $\encRow{}$.
    Fix an arbitrary $i$.
    Again due to \cref{condition:rows-of-cells-structure}, $v_i$ has the form
    \begin{align*}
        \textstyle\biglanguageconcat_{1 \le j \le w_i}\bracket{\wCB \concat \wPref*{p_{i,j}} \concat t_{i,j} \concat \wSuf*{s_{i,j}} \concat \wA \concat \wCE}\,,
    \end{align*}
    where $w_i \in \N$, $p_{i,j}, s_{i,j} \in N_n^*$, $p_{i,1} = s_{i,w_i} = n$, and $t_{i,j} \in \cT$.
    Due to \cref{cond:cells-prefix-and-first-letters-of-suffixes-form-a-comb}, we have that all $s_{i,j}$ are nonempty and
    \begin{align}
        \wPref*{p_{i,1}} \concat \wSuf*{ s_{i,1}[1] \concat s_{i,2}[1] \concat s_{i,3}[1] \cdots s_{i,w_i}[1]} = \comb{n}\,. \label{eq:letters-form-comb}
    \end{align}
    This implies that $w_i = 2^n$.
    By \cref{cond:cells-prefix-and-first-letters-of-suffixes-form-a-comb,eq:letters-form-comb} we get that
        $p_{i,j}  = \comb{n}[1,j]$, 
    and now \cref{cond:cells-contain-combs} implies that $s_{i,j} = \comb{n}[j+1,2^n+1]$, so $\wRB \concat v_i \concat \wRE$ is a valid encoding of a row of length $2^n$.
    Hence $u$ encodes a tiling $T \coloneqq \myMatrix{t_{i,j}}_{i,j} \in \cT^{h \times 2^n}$.
    Property~\eqref{eq:tiling-property-b} in the definition of a valid tiling is now trivially implied by \cref{cond:colours-match-horizontally}, and we only need to show~\cref{eq:tiling-property-a}.
    Fix an arbitrary pair of tiles $t_{i,j}, t_{i+1,j}$ that are vertical neighbours.
    Observe that $p_{i,j} s_{i+1,x} = \comb{n} \iff x = j$.
    Therefore, by \cref{cond:colours-match-vertically}, $\tBottom{t_{i,j}} = \tTop{t_{i+1,j}}$, thus $T$ is a valid $\cT$-tiling, and, by \cref{condition:corners-fixed}, a valid $(\cT, \toplefttile, \bottomrighttile)$-tiling.
\end{proof}

        \subsection[Construction of the automaton A\_I]{Construction of the automaton \texorpdfstring{\boldmath$\cA_{\I}$}{A\_I}} \label{subsec:construction-of-the-automaton}
        Let $\Sigma_{\I,\wSharp} \coloneqq \Sigma_\I \cup \set{\wSharp}$.
Here, we define the \NFA $\cA_\I$ over $\Sigma_{\I,\wSharp}^2$ and prove 
\typeout{HEREABCQ}\cref{lem:projection-of-relation-equals-valid-encodings}, which states that $\univProj[1]{\lang{\cA_\I}} = \encLang$. 
The construction we present in this section, however, does not require the full generality of the setting of automatic structures:
\begin{itemize}
    \item $\automatonToRelation{\cA_\I}$ only holds for words of the same length, \ie, $\cA_\I$ rejects words with $\wSharp$;
    \item we only use a subset of the alphabet: $\Sigma_\I \times N_n \subset \Sigma_{\I,\wSharp}^2$.
\end{itemize}
For this reason, we begin with a lemma, which allows us to focus only on words satisfying the above properties.
Let $\rhoProj{}\colon (\Sigma_\I \times N_n)^* \to \Sigma_\I^*$ be a homomorphism given by $\rhoProj{a, \argumentDot} \coloneqq a$.
Additionally, let
\begin{align*}
    \langForall{L} &\coloneqq \set{ w \in \Sigma_\I^* \suchthat \rhoInv{\hiddenset{w}} \subseteq L}\,.
\end{align*}

\begin{lem}
    \label{lem:simplification-of-construction}
    For any \NFA $\cA'_\I$ over $\Sigma_\I \times N_n$
    with $\BigO{n\cdot \size{\cT}}$ states,
    there exists an \NFA $\cA_\I$ over $\Sigma_{\I,\wSharp}^2$ with $\BigO{n\cdot \size{\cT}}$ states
    such that $\univProj[1]{\lang{\cA_\I}} = \langForall{\lang{\cA'_\I}}$.
\end{lem}

\begin{proof}
    Take any $\cA'_\I$ over $\Sigma_\I \times N_n$.
    Let
    \begin{align*}
        \cE_1 &\coloneqq \bracket{\Sigma_\I^2}^* \bracket{\Sigma_\I \times \set{\wSharp}}^+ + \bracket{\Sigma_\I^2}^* \bracket{\set{\wSharp} \times \Sigma_\I}^+ & &\text{\scriptsize($u \otimes v$ such that $\size u \neq \size v$)}\\
        \cE_2 &\coloneqq \bracket{\Sigma_\I^2}^* \bracket{\Sigma_\I \times \bracket{\Sigma_\I \setminus N_n}} \bracket{\Sigma_\I^2}^* & &\text{\scriptsize(words with letter from $\Sigma_\I^2 \setminus (\Sigma_\I \times \N_n$))}\\
        \cA_\I &\coloneqq \cA'_\I \oplus \regexToAutomaton{\cE_1} \oplus \regexToAutomaton{\cE_2}\,.
    \end{align*}
    By definition, a word $w$ belongs to $\univProj[1]{\lang{\cA_\I}}$ whenever for all $v$ the word $w \otimes v$ belongs to $\lang{\cA_\I}$.
    By construction, $\lang{\cA_\I}$ contains all $w \otimes v$ where $\size{v} \neq \size{w}$ ($\cE_1$) or where $v$ is using a symbol from $\Sigma_\I \setminus N_n$ ($\cE_2$).
    Hence, the only words which can be missing from $\lang{\cA_\I}$ come from $\lang{\cA'_\I}$.
    This implies that $\univProj[1]{\lang{\cA_\I}} = \langForall{\lang{\cA'_\I}}$.
\end{proof}
Therefore, we only have to provide $\cA'_\I$ such that $\langForall{\lang{\cA'_\I}} = \encLang$.
The construction is modular, based on six \NFA corresponding to~\crefrange{condition:first}{condition:last}:

\begin{lem}[modular design]
    \label{lem:construction-of-conjunction}
    For any six \NFA $(\cC^i_\I)_{1 \le i \le 6}$ 
    over $\Sigma_\I \times N_n$
    with $\BigO{n\cdot \size{\cT}}$ 
    states,
    there exists an \NFA $\cA'_\I$ 
    with $\BigO{n \cdot \size{\cT}}$ states
    over $\Sigma_\I \times N_n$ such that
    \begin{align*}
        \langForall{\lang{\cA'_\I}} = \wA \! \bigcap_{1 \leq i \leq 6} \langForall{\lang{\cC^i_\I}}\,.
    \end{align*}
\end{lem}

\begin{proof}
    Define
    \begin{align*}
        \cH &\coloneqq \bracket{\set{\wA} \times N_n \setminus \set{\w1, \w2, \dots, \w6}} \concat \bracket{\Sigma_\I \times N_n}^* \\
        \cA'_\I &\coloneqq
        \regexToAutomaton{(\wA, \w1)} \odot \cC^1_\I \ \ \oplus \ \ \regexToAutomaton{(\wA, \w2)} \odot \cC^2_\I \ \ \oplus \ \ \cdots \ \ \oplus \ \ \regexToAutomaton{(\wA, \w6)} \odot \cC^6_\I \ \ \ \oplus \ \ \ \regexToAutomaton{\cH}\,.
    \end{align*}
    Observe that
    \begin{align*}
        \wA w \in \langForall{\lang{\cA'_\I}}
        &\iff
        \rhoInv{\hiddenset{\wA w}} \subseteq \lang{\cA'_\I}
        \iff
        (\set{\wA} \times N_n) \concat \rhoInv{\hiddenset{w}} \subseteq \lang{\cA'_\I} \iff {} \\
        &\iff \forall{i \in N_n}. (\wA, i) \concat \rhoInv{\hiddenset{w}} \subseteq \lang{\cA'_\I}\,,
    \end{align*}
    but trivially
    \begin{align*}
        \lang{\regexToAutomaton{(\wA, j)} \odot \cC^j_\I} \cap (\wA, i) \concat \rhoInv{\hiddenset{w}} &= \emptyset & &\text{for any $1\le i,j\le 6$, $i \neq j$} \\
        \lang{\regexToAutomaton{\cH}} \cap (\wA, i) \concat \rhoInv{\hiddenset{w}} &= \emptyset & &\text{for any $1\le i\le 6$.}
    \end{align*}
    Therefore, $\wA w \in \langForall{\lang{\cA'_\I}} $ if, and only if, $\rhoInv{\hiddenset{w}} \subseteq \langForall{\lang{\cC^i_\I}}$ for all $1\le i \le 6$.
\end{proof}
By the definition of $\encLang$, it only remains to construct automata $\cC^i_\I$ such that $\langForall{\lang{\cC^i_\I}} = \condLang{i}$ for $1 \le i \le 6$.
The construction is easy for~\crefrange{condition:rows-of-cells-structure}{cond:colours-match-horizontally}:
\begin{align*}
    \cC^i_\I &\coloneqq \rhoInv{\cB^i_\I} & &\text{for $i \in \set{1,2,3}$}
\end{align*}
as $\langForall{\lang{\rhoInv{\cA}}} = \lang{\cA}$ for any \NFA $\cA$.
Observe that the number of states of all $\cC_\I^i$ is upper bound
by $\BigO{\cT}$.
The remaining~\crefrange{cond:cells-contain-combs}{cond:colours-match-vertically} all talk about ``every proper suffix'' satisfying some simple regular property.
We handle that in a general way.
For $L \subseteq \bracket{\Sigma_\I \times N_n}^*$, define
\begin{align*}
    \allSufLang{L} \coloneqq \set{ w \suchthat v \in \langForall{L} \text{ for all proper suffixes }  v \text{ of } w }
\end{align*}

\begin{lem}[recognising ``for all proper suffixes'']
    For any \NFA $\cA$ over $\Sigma_\I \times N_n$ with $k$ states, there exists an \NFA $\allSufAut{\cA}$ with $\BigO{k}$ states such that
    \begin{align*}
        \langForall{\lang{\allSufAut{\cA}}} = \allSufLang{\lang{\cA}}\,.
    \end{align*}
\end{lem}
\begin{proof}
    Fix any \NFA $\cA = (Q, \Sigma_\I \times N_n, \delta, Q_\ini, Q_\fin)$.
    We define $\allSufAut{\cA}$ which guesses the suffix to verify
    \begin{align*}
        \allSufAut{\cA} \coloneqq (Q \cup \set{s}, \Sigma_\I \times N_n, \delta \cup \delta', \set{s}, Q_\fin \cup \set{s})
    \end{align*}
    for some fresh state $s \notin Q$, and $\delta'$ containing transitions
    \smash{$s \xrightarrow{(a, 0)} s$ for $a \in \Sigma_\I$ and
        $s \xrightarrow{(a, i)} q$} for $a \in \Sigma_\I$, $i \in N_n^{>0}$, $q \in Q_\ini$.
    Also, let $\tau$ be the homomorphism such that $\tau(a) \coloneqq (a, 0)$.

    \proofCase{Inclusion ``$\subseteq$''}
    Take any $w \in \langForall{\lang{\allSufAut{\cA}}}$.
    Let $v$ be any proper suffix of $w$.
    Take any $v' \in \rhoInv{\hiddenset{v}}$.
    We need to show that $v' \in \lang{\cA}$.
    The word $w$ can be written as $u a v$, for $\size{u} \ge 0$ and $\size{a} = 1$.
    Consider a word $w' = \tau(u) (a, 1) v'$.
    By definition of $\langForall{}$, $w' \in \lang{\allSufAut{\cA}}$.
    Let $r$ be an accepting run of $\allSufAut{\cA}$ over $w'$.
    By construction, the run stays in state $s$ while reading $\tau(u)$ and goes to some $q \in Q_\ini$ upon reading $(a,1)$.
    Therefore, the remaining suffix of $r$ is an accepting run of $\cA$ over $v'$.

    \proofCase{Inclusion ``$\supseteq$''}
    Fix $w \in \allSufLang{\lang{\cA}}$.
    Take any $w' \in \rhoInv{\hiddenset{w}}$.
    We will show that $w' \in \lang{\allSufAut{\cA}}$.
    Let $u' (a, k) v' \coloneqq w'$ be such that $u'$ is the maximal prefix arising as $\tau(u)$ for some $u$ (possibly empty).
    Note that $k \neq 0$.
    By assumption, $v' \in \lang{\cA}$, so there exists an accepting run $r_2$ of $\cA$ over $v'$ starting in some $q \in Q_{\ini}$.
    By construction, there exists a run $r_1$ from $s$ to $q$ over $u' (a, k)$ in $\allSufAut{\cA}$.
    Hence, the run $r_1 r_2$ accepts $w'$.
\end{proof}
To handle conditions ``beginning with $\wCB$'' and ``containing $\wRB$'' appearing as antecedents of implications, we proceed in the vein of the equivalence $a \rightarrow b \equiv \neg a \lor b$.
Let
\begin{align*}
    \cG_{\neg\wCB} &\coloneqq \bracket{\Sigma_{\I} \setminus \set{\wCB}} \concat \Sigma_{\I}^* &
    \cG_{\neg\wRB} &\coloneqq \bracket{\Sigma_{\I} \setminus \set{\wRB}}^*\,.
\end{align*}

\begin{lem}
    \label{lem:suffices-to-define-inner-condition}
    For any $i\in \set{4,5,6}$, given an \NFA $\hat\cC^i_\I$
    of size $O\big(n \cdot \size{\cT})$ and satisfying \[\langForall[][\big]{\lang[][][\big]{\hat\cC^i_\I}} = \set{ w \suchthat \comb{n} \concat \wA \in \cF^i_\I(w) }\,,\] one can construct $\cC^i_\I$ with $O\big(n \cdot \size{\cT})$ 
    states such that
    $
    \langForall{\lang{\cC^i_\I}} = \condLang{i}
    $.
\end{lem}
\begin{proof}
    Fix $\hat\cC^4_\I, \hat\cC^5_\I, \hat\cC^6_\I$ as in the statement of the lemma.
    We define $\cC^i_\I$ as
    \begin{align*}
        \cC^4_\I &\coloneqq \allSufAut[][][\big]{\hat\cC^4_\I \oplus \rhoInv[\big]{\regexToAutomaton{\cG_{\neg\wCB}}}} \\
        \cC^5_\I &\coloneqq \allSufAut[][][\big]{\hat\cC^5_\I \oplus \rhoInv[\big]{\regexToAutomaton{\cG_{\neg\wCB}}}} \\
        \cC^6_\I &\coloneqq \allSufAut[][][\big]{\hat\cC^6_\I \oplus \rhoInv[\big]{\regexToAutomaton{\cG_{\neg\wCB} + \cG_{\neg\wRB}}}}\,.
    \end{align*}
    The above cases are similar;
    \abbrev{w.l.o.g.} let us focus on $\cC^4_\I$.
    Observe that
    \begin{align*}
        \langForall[][\big]{
            \lang[][][\big]{\hat\cC^4_\I \oplus \rhoInv[\big]{\regexToAutomaton{\cG_{\neg\wCB}}}}
        }
        =
        \langForall[][\big]{
            \lang[][][\big]{\hat\cC^4_\I}
        }
        \cup
        \lang[][][\big]{
            \cG_{\neg\wCB}
        }
        =
        \set{ w \suchthat \comb{n} \concat \wA \in \cF^4_\I(w) } \cup
        \lang[][][\big]{
            \cG_{\neg\wCB}
        }\,,
    \end{align*}
    which directly corresponds to \cref{cond:cells-contain-combs}, as required.
    Observe that the arguments to $\allSufAut{\cdot}$ are \NFA
    with $O(n\cdot \size{\cT})$ states.
\end{proof}
The essential element needed to define \NFA $\hat\cC^i_\I$ as in \cref{lem:suffices-to-define-inner-condition} is an \NFA for the language $\set{ \comb{n} \wA }$.
First, we define $\comb{n}$ as an intersection of languages of $n+1$ regular expressions. We then show how that intersection can be concisely represented by an \NFA $\cC_n$ with $\BigO{n}$ states such that 
    $\langForall{\lang{\cC_n}} = \set{ \comb{n} \wA }$. 

\begin{defi}
    We define $n+1$ regular expressions $\cE_i$ over $\Sigma_\I$
    \begin{align*}
        \cE_0 &\coloneqq N^{>1}_n \concat \bracket{\w0 \concat N^{>1}_n}^* \\
        \cE_i &\coloneqq N^{>i}_n \concat \bracket{\bracket{N^{<i}_n}^* \concat i \concat \bracket{N^{<i}_n}^* \concat N^{>i}_n}^*
        &
        \text{ for $0 < i < n$}\\
        \cE_n &\coloneqq n \concat \bracket{N^{<n}_n}^* \concat n\,.
    \end{align*}
\end{defi}
\begin{lem}
    \label{lem:comb-singleton-is-the-intersection-of-Es}
    $\set{ \comb{n} } = \bigcap_{0 \le i \le n} \lang{\cE_i}$.
\end{lem}
\begin{proof}
    It is easy to prove the inclusion ``$\subseteq$'' by unravelling the definition of $\comb{n}$.

    \proofCase{Inclusion ``$\supseteq$''}
    Take any $w \in \bigcap_{0 \le i \le n} \lang{\cE_i}$.
    We will show that $w = \comb{n}$.
    \begin{clm}
        For $0 \le k \le n-1$, we have $\bigcap_{0 \le i \le k} \lang{\cE_i} = \lang{N_n^{>k} \concat \bracket{\combInner{k} \concat N_n^{>k}}^*}$.
    \end{clm}
    We prove the claim by induction.
    The base case is trivial.
    Fix a word
    \begin{align*}
        w \in \lang{N_n^{>k} \concat \bracket{\combInner{k} \concat N_n^{>k}}^*} \cap \lang{\cE_{k+1}}\,.
    \end{align*}
    It has the form $w = a_1 \concat \combInner{k} \concat a_2 \concat \combInner{k} \concat \cdots \concat  \combInner{k} \concat a_m$ for some $m \ge 2$ and $a_1, a_2, \ldots, a_m \in N_n^{>k}$.
    But since $w \in \lang{\cE_{k+1}}$, 
    all of $a_{2}, a_4, a_6, \ldots$ are equal to $k+1$, and $m$ is odd.
    Thus
    \begin{align*}
        w = a_1 \concat \underbrace{\combInner{k} \concat (k+1) \concat \combInner{k}}_{\combInner{k+1}} \concat \cdots \concat  \underbrace{\combInner{k} \concat (k+1) \concat \combInner{k}}_{\combInner{k+1}} \concat a_m
    \end{align*}
    We conclude by noticing that $\lang{\cE_n} \cap \lang{N_n^{>(n-1)} \concat \bracket{\combInner{n-1} \concat N_n^{>(n-1)}}^*} = \set{\comb{n}}$.
\end{proof}
Let us define
\[
    \cC_n \coloneqq \rhoInv{\regexToAutomaton{\cE_0}} \odot \regexToAutomaton{(\wA,\w0)} \oplus
    \rhoInv{\regexToAutomaton{\cE_1}} \odot \regexToAutomaton{(\wA,\w1)} \oplus \dots \oplus
    \rhoInv{\regexToAutomaton{\cE_n}} \odot \regexToAutomaton{(\wA,n)}\,.
\]
Observe that each $\regexToAutomaton{\cE_i}$ and $\regexToAutomaton{(\wA,\w i)}$ has a constant number of states, hence by further applying \cref{fact:inv-hom}, $\cC_n$ has $O(n)$ states.

\begin{lem}
    \label{lem:universal-projection-of-Cn-is-the-comb}
    $
    \langForall{\lang{\cC_n}} = \bracket{\bigcap_{0 \le i \le n} \lang{\cE_i}} \concat \wA
    $.
\end{lem}

\begin{proof}
    \proofCase{Inclusion ``$\subseteq$''}
    Take any $w = u \wA \in \langForall{\lang{\cC_n}}$,
    and $i \in N_n$.
    We prove that $u \in \lang{\cE_i}$.
    By definition, $\rhoInv{\hiddenset{u \wA}} \subseteq \lang{\cC_n}$.
    Fix a homomorphism $\tau(a) = (a, 0)$.
    Note that $\tau(u) (\wA, i) \in \lang{\cC_n}$.
    This can be accepted only by the $\rhoInv{\regexToAutomaton{\cE_i}} \odot \regexToAutomaton{(\wA,i)}$ component,
    thus $u \in \lang{\regexToAutomaton{\cE_i}} = \lang{\cE_i}$, as required.

    \proofCase{Inclusion ``$\supseteq$''}
    Take any $w = u \wA \in \bracket[\big]{\bigcap_{0 \le i \le n} \lang{\cE_i}} \concat \wA$.
    Take any $u' (\wA, i) \in \rhoInv{\hiddenset{u\wA}}$.
    Since $u \in \lang{\cE_i}$, $u' \in \rhoInv{\lang{\cE_i}}$, and $u' (\wA, i) \in \lang{\rhoInv{\regexToAutomaton{\cE_i}}\odot \regexToAutomaton{(\wA, i)}}$, as required.
\end{proof}

\begin{defi}[\NFA $\hat\cC_\I^i$]
Let $\I$ be of width $2^n$ and let $i \in \set{4,5,6}$. Take the \NFA $\cC_n = \tuple{Q^{(1)},\Sigma_\I \times N_n,\delta^{(1)},Q_\ini^{(1)},Q_\fin^{(1)}}$ as above with $\BigO{n}$ states
and $\regexToAutomaton{\cF_\I^i} = \tuple{Q^{(2)},\Sigma_\I \times \Phi,\delta^{(2)},Q_\ini^{(2)},Q_\fin^{(2)}}$ with $\BigO{\size{\cT}}$ states.
We define $\hat\cC_\I^i \coloneqq \tuple{Q,\Sigma_\I \times N_n,\delta,Q_\ini,Q_\fin}$ with $\BigO{n\cdot \size{\cT}}$ states, where
\begin{align*}
    Q &\coloneqq Q^{(1)} \times Q^{(2)}, &
    Q_{\ini} &\coloneqq Q_\ini^{(1)} \times Q_\ini^{(2)}, &
    Q_{\fin} &\coloneqq Q_\fin^{(1)} \times Q_\fin^{(2)},
\end{align*}
and the transition relation is
\begin{align*}
    \delta \coloneqq {} &\set[\Big]{ (p, q) \xrightarrow{(a, \alpha)} (r, s) \suchthat q \xrightarrow{(a,\top)} s \in \delta^{(2)} \land p \xrightarrow{(a, \alpha)} r \in \delta^{(1)} } \cup {} \\
    &\set[\Big]{ (p, q) \xrightarrow{(a, \alpha)} (p, s) \suchthat q \xrightarrow{(a,\bot)} s \in \delta^{(2)} \land p \in Q^{(1)} }\,.
\end{align*}
\end{defi}
Intuitively, $\hat\cC_\I^i$ runs $\cC_n$ over the fragments of the input which are underlined by $\cF_\I^i$.

\begin{fact} \label{fact:complicated-language-membership}
    We have $w \in \lang[][][\big]{\hat\cC_\I^i}$ if, and only if, there is some ${v \in \lang{\rhoInv{\cF^i_\I}}}$ such that $\inProj{v} = w$ and $\outProj{v} \in \lang{\cC_n}$.
\end{fact}

To finish the construction, we need to prove the following property:
\begin{lem}
    For $i \in \set{4,5,6}$
    \begin{align*}
        \langForall[][\big]{\lang[][][\big]{\hat\cC_\I^i}} = \set{ w \suchthat \comb{n} \concat \wA \in \cF^i_\I(w) }\,.
    \end{align*}
\end{lem}
As the proofs for $i \in \set{4,5,6}$ are analogous, we focus on the hardest one,
and then only comment how it can be adapted for $i \in \set{4,5}$.
\begin{proof}[Proof \textup{($i = 6$)}]
    \proofCase{Part A: inclusion ``$\subseteq$''}
    Let $w \in \langForall[][\big]{\lang[][][\big]{\hat\cC_\I^6}}$.
    Define
    \begin{align*}
        U &\coloneqq \set{ u \in \lang{\cF^6_\I} \suchthat \inProj{u} = w} \subseteq \bracket{\Sigma_\I \times \Phi}^*\,.
    \end{align*}
    Note that if $U = \emptyset$, then $\cF^6_\I(w) = \emptyset$. Hence, by \cref{fact:complicated-language-membership}, $\lang[][][\big]{\hat\cC_\I^6} = \emptyset$, and $\langForall[][\big]{\lang[][][\big]{\hat\cC_\I^6}} = \emptyset$, a contradiction.
    Therefore, $U \neq \emptyset$, and $w \in \lang{\inProj{\cF^6_\I}}$, so it is of the form
    \begin{align*}
        \concat \wCB \concat \wPref{p} \concat t \concat \wSuf{s} \concat \wA \concat \wCE \concat \beta \concat \wRE \concat \wRB
        \concat \wCB \concat \wPref{p_1} \concat t_1 \concat \wSuf{s_1} \concat \wA \concat \wCE
        \concat \wCB \concat \wPref{p_2} \concat t_2 \concat \wSuf{s_2} \concat \wA \concat \wCE
        \concat \cdots
        \concat \wCB \concat \wPref{p_k} \concat t_k \concat \wSuf{s_k} \concat \wA \concat \wCE \concat \wRE \concat \gamma
    \end{align*}
    for some $k \in \N$, $p, p_i, s, s_i \in N_n^*$, $t, t_i \in \cT$, $\beta \in \bracket{\Sigma_\I \setminus \set{\wRB, \wRE}}^*$ and $\gamma \in \Sigma_\I^*$.
    Furthermore, $\size{U} = k$ and $U$ contains the following underlined words $u^{(6)}_1, \ldots, u^{(6)}_k \in \bracket{\Sigma_\I \times \Phi}^*$:
    \begin{align*}
        u^{(6)}_1 &= \wCB \concat \wOutD{\wPref{p}} \concat t \concat \wSuf{s} \concat \wA \concat \wCE \concat \beta \concat \wRE \concat \wRB
        \concat \wCB \concat \wPref{p_1} \concat t_1 \concat \wOutD{\wSuf{s_1} \concat \wA} \concat \wCE
        \concat \wCB \concat \wPref{p_2} \concat t_2 \concat \wSuf{s_2} \concat \wA \concat \wCE
        \concat \cdots
        \concat \wCB \concat \wPref{p_k} \concat t_k \concat \wSuf{s_k} \concat \wA \concat \wCE \concat \wRE \concat \gamma \\
        u^{(6)}_2 &= \wCB \concat \wOutD{\wPref{p}} \concat t \concat \wSuf{s} \concat \wA \concat \wCE \concat \beta \concat \wRE \concat \wRB
        \concat \wCB \concat \wPref{p_1} \concat t_1 \concat \wSuf{s_1} \concat \wA \concat \wCE
        \concat \wCB \concat \wPref{p_2} \concat t_2 \concat \wOutD{\wSuf{s_2} \concat \wA} \concat \wCE
        \concat \cdots
        \concat \wCB \concat \wPref{p_k} \concat t_k \concat \wSuf{s_k} \concat \wA \concat \wCE \concat \wRE \concat \gamma \\
        &\vdots \\
        u^{(6)}_k &= \wCB \concat \wOutD{\wPref{p}} \concat t \concat \wSuf{s} \concat \wA \concat \wCE \concat \beta \concat \wRE \concat \wRB
        \concat \wCB \concat \wPref{p_1} \concat t_1 \concat \wSuf{s_1} \concat \wA \concat \wCE
        \concat \wCB \concat \wPref{p_2} \concat t_2 \concat \wSuf{s_2} \concat \wA \concat \wCE
        \concat \cdots
        \concat \wCB \concat \wPref{p_k} \concat t_k \concat \wOutD{\wSuf{s_k} \concat \wA} \concat \wCE \concat \wRE \concat \gamma\,.
    \end{align*}
    We distinguish two cases based on whether some $u\in U$ satisfies $\outProj{\rhoInv{\hiddenset{u}}} \subseteq \lang{\cC_n}$.
    
    \proofCase{Case A.1: such a $u$ exists}
    Take any such $u \in U$.
    Observe that $\outProj{u} \in \langForall{\lang{\cC_n}} = \set{ \comb{n} \wA }$.
    Hence, $\inProj{u} = w$, $\outProj{u} = \comb{n} \wA$, and $u \in \lang{\cF_\I^6}$.
    Therefore, $\comb{n} \wA \in \cF_\I^6(w)$, as required.

    \proofCase{Case A.2: such a $u$ does not exist}
    Therefore, for every $u \in U$, there exists $v_u \in \rhoInv{\hiddenset{u}}$ such that $\outProj{v_u} \notin \lang{\cC_n}$.
    Fix any family $(v_u)_{u \in U}$ of such words.
    Let $\alpha_u$ be the position of the last underlined symbol in $u$.
    We now define a word $w' \in \rhoInv{\hiddenset{w}}$ by choosing ``problematic'' output projections  at the positions $\alpha_u$. Formally, for every position $i$ in $w$ let
    \begin{align*}
        w'[i] \coloneqq \begin{cases}
                    \outProj{v_u[i]} & \text{if $i = \alpha_u$ for some $u$} \\
        (w[i], 0) & \text{otherwise}\,.
        \end{cases}
    \end{align*}
    Observe that $w'$ is properly defined, as positions $\alpha_u$ are pairwise different (corresponding to the letters $\wA$ following blocks $s_1$ to $s_k$).
    Since $\rhoProj{w'} = w$ and $w \in \langForall[][\big]{\lang[][][\big]{\hat\cC_\I^6}}$, we conclude that $w' \in \lang[][][\big]{\hat\cC_\I^6}$.
    By \cref{fact:complicated-language-membership}, we obtain $v \in \lang{\rhoInv{\cF^6_\I}}$ such that $\inProj{v} = w'$ and $\outProj{v} \in \lang{\cC_n}$.
    However, $\rhoProj{\outProj{v}} = \rhoProj{\outProj{v_u}}$ for some $u \in U$ and the last symbols of $\outProj{v}$ and $\outProj{v_u}$ are identical.
    Since by construction $\cC_n$ ignores the component $N_n$ of its alphabet $\Sigma_I \times N_n$ for all letters but the last one, we get that
    \begin{align*}
        \outProj{v} \in \lang{\cC_n} \quad \text{if and only if} \quad \outProj{v_u} \in \lang{\cC_n}\,.
    \end{align*}
    This contradicts the fact that $\outProj{v} \in \lang{\cC_n}$ and completes the analysis of Part A.

    \proofCase{Part B: inclusion ``$\supseteq$''}
    Consider $w$ such that $\comb{n}\wA \in \cF_\I^6(w)$.
    From the definition of $\cF_\I^6(w)$, we can fix $v \in \lang{\cF_\I^6}$ such that $\inProj{v} = w$ and $\outProj{v} = \comb{n}\wA$.
    We need to show that $\rhoInv{\hiddenset{w}} \subseteq \lang[][][\big]{\hat\cC_\I^6}$.
    Take any $w' \in \rhoInv{\hiddenset{w}}$. By 
    \cref{fact:complicated-language-membership}, this
    requires 
    constructing $u \in \lang{\rhoInv{\cF_\I^6}}$ such that $\inProj{u} = w'$ and $\outProj{u} \in \lang{\cC_n}$. We do this by combining $w'$ and $v$.
    Let $u \in \bracket{\Sigma_\I \times \N_n \times \Phi}^*$ be the unique word such that $\inProj{u} = w'$ and for every position $i$ in $u$ we have $u[i] = (\_, \_, \top)$ if and only if $v[i] = (\_, \top)$.
    We have to show that $\outProj{u} \in \lang{\cC_n}$. Note that \[\outProj{u} \in \rhoInv{\comb{n}\wA} = \rhoInv{\Big(\bigcap\nolimits_{0 \le i \le n} \lang{\cE_i}\Big)\wA} = \rhoInv{\langForall{\lang{\cC_n}}}\,,
    \]
    where the first inclusion follows from the definition of $u$, while the
    subsequent equalities follow from \cref{lem:comb-singleton-is-the-intersection-of-Es,lem:universal-projection-of-Cn-is-the-comb}, respectively. We get that $\rho_\I(\outProj{u}) \in \langForall{\lang{\cC_n}}$, and thus $\outProj{u} \in \rhoInv{\rho_\I(\outProj{u})} \subseteq \lang{\cC_n}$.
\end{proof}

\begin{proof}[Proof \textup{($i \in \set{4,5}$)}]
    The proof is analogous to the case $i = 6$.
    As the cases are distinguished by the filter $\cF_\I^i$ being used, the only differences are related to the shape of words matched by $\inProj{\cF_\I^i}$.
    In particular, the set $U$ for $i \in \set{4,5}$ is now a singleton containing $u^{(i)}$:
    \begin{align*}
        u^{(4)} &= \wCB \concat \wOutD{\wPref{p}} \concat t \concat \wOutD{\wSuf{s} \concat \wA} \concat \wCE \concat \gamma \tag{$i = 4$} \\
        u^{(5)} &= \concat
        \wCB \concat \wOutD{\wPref{p_1}} \concat t \concat \wSuf{\wOutD{s'_1} \concat s_1} \concat \wA \concat \wCE \concat
        \wCB \concat {\wPref{p_2}} \concat t \concat \wSuf{\wOutD{s'_2} \concat s_2} \concat \wA \concat \wCE \cdots
        \wCB \concat {\wPref{p_{k-1}}} \concat t \concat \wSuf{\wOutD{s'_{k-1}} \concat s_{k-1}} \concat \wA \concat \wCE \concat
        \wCB \concat {\wPref{p_k}} \concat t \concat \wOutD{\wSuf{s'_k} \concat \wA} \concat \wCE \concat \wRE \concat \gamma \tag{$i = 5$}
    \end{align*}
    The rest of the proof only requires substituting $\cF_\I^6$ with $\cF_\I^4$ or $\cF_\I^5$.
\end{proof}

    \section{NFA of doubly exponential size after universal projection}\label{sec:size-of-automaton}
    \newcommand{\tilingPictureCounterFixed}{
    \begin{tikzpicture}
        \matrix[left delimiter={[},right delimiter={]},row sep=2pt,column sep=2pt,inner xsep=0,ampersand replacement=\&]
        {
            \tileInner{2}{2}{5}{5} \& \tileInner{2}{2}{0}{2} \& \tileInner{2}{2}{0}{2} \& \tileInner{2}{2}{0}{2} \& \tileInner{2}{3}{3}{2} \\
            \tileInner{5}{0}{5}{5} \& \tileInner{0}{0}{0}{0} \& \tileInner{0}{0}{0}{0} \& \tileInner{0}{1}{1}{0} \& \tileInner{3}{3}{3}{1} \\
            \tileInner{5}{0}{5}{5} \& \tileInner{0}{0}{0}{0} \& \tileInner{0}{1}{1}{0} \& \tileInner{1}{1}{0}{1} \& \tileInner{3}{3}{3}{1} \\
            \tileInner{5}{0}{5}{5} \& \tileInner{0}{0}{0}{0} \& \tileInner{1}{0}{1}{0} \& \tileInner{0}{1}{1}{0} \& \tileInner{3}{3}{3}{1} \\
            \tileInner{5}{0}{5}{5} \& \tileInner{0}{1}{1}{0} \& \tileInner{1}{1}{0}{1} \& \tileInner{1}{1}{0}{1} \& \tileInner{3}{3}{3}{1} \\
            \tileInner{5}{0}{5}{5} \& \tileInner{1}{0}{1}{0} \& \tileInner{0}{0}{0}{0} \& \tileInner{0}{1}{1}{0} \& \tileInner{3}{3}{3}{1} \\
            \tileInner{5}{0}{5}{5} \& \tileInner{1}{0}{1}{0} \& \tileInner{0}{1}{1}{0} \& \tileInner{1}{1}{0}{1} \& \tileInner{3}{3}{3}{1} \\
            \tileInner{5}{0}{5}{5} \& \tileInner{1}{0}{1}{0} \& \tileInner{1}{0}{1}{0} \& \tileInner{0}{1}{1}{0} \& \tileInner{3}{3}{3}{1} \\
            \tileInner{5}{4}{4}{5} \& \tileInner{1}{4}{4}{4} \& \tileInner{1}{4}{4}{4} \& \tileInner{1}{4}{4}{4} \& \tileInner{3}{3}{4}{4} \\
        };
    \end{tikzpicture}
}

From the lower bounds established in
    \Cref{subsec:construction-of-the-automaton}, it is now easy to
    construct a family $(\cA_{n})_{n \in \N}$ of \NFA,
    each of size $\BigO{n^4}$, such that
    the smallest \NFA after a universal projection step has
    has a doubly exponential number of states. 
    This is achieved by simulating a binary counter over an exponential number of bits in a proper tiling.

    To this end, for a fixed $n$, let $\I_n$ be the problem instance from the
    left-hand side of \cref{fig:problem-instance-and-solution}, and define $\cA_n \coloneqq \cA_{\I_n}$. Intuitively, the colours $0, 1$ vertically represent the counter bits, and horizontally encode the carry-over bit.
    The only valid $(\cT_\increment, \toplefttile, \bottomrighttile)$-tiling of width $k>2$ simulates incrementing an
    $(k-2)$-bit binary counter from $0$ to $2^{k-2}-1$; see
    the right-hand side of 
    \Cref{fig:problem-instance-and-solution} for an
    example for width $k=5$.
    Thus, after a universal projection step, since $\cA_n$ enforces width
    $k=2^n$, the     resulting \NFA accepts a single word of length doubly exponential in $n$.

\begin{prop}
    \label{prop:the-size-of-the-NFA}
    The \NFA for $\univProj[1][\big]{\lang[][][\big]{\cA_{(\cT_\increment, \toplefttile, \bottomrighttile, n)}}}$ has size $\Omega\left(2^{2^n}\right)$.
\end{prop}

\begin{figure}[b]
\centering \hfill
\begin{minipage}[m]{0.5\textwidth}
    \centering
    \begin{align*}
        \cT_\increment &\coloneqq
            \begin{tikzpicture}[baseline=-0.53ex]
                \matrix[left delimiter={\{},right delimiter={\}},row sep=2pt,column sep=2pt,inner xsep=0,ampersand replacement=\&]
                {
                    \node[inner sep=0]{\tikz[baseline=-0.53ex]{\tileInner{2}{2}{5}{5}},}; \& \node[inner sep=0]{\tikz[baseline=-0.53ex]{\tileInner{2}{2}{0}{2}},}; \& \node[inner sep=0]{\tikz[baseline=-0.53ex]{\tileInner{2}{3}{3}{2}},}; \\
                    \node[inner sep=0]{\tikz[baseline=-0.53ex]{\tileInner{5}{0}{5}{5}},};, \& \ \& \node[inner sep=0]{\tikz[baseline=-0.53ex]{\tileInner{3}{3}{3}{1}},}; \\
                    \node[inner sep=0]{\tikz[baseline=-0.53ex]{\tileInner{5}{4}{4}{5}},}; \& \node[inner sep=0]{\tikz[baseline=-0.53ex]{\tileInner{1}{4}{4}{4}},}; \& \node[inner sep=0]{\tikz[baseline=-0.53ex]{\tileInner{3}{3}{4}{4}}\phantom{,}}; \\
                };
            \end{tikzpicture} \cup
        \begin{tikzpicture}[baseline=-0.53ex]
            \matrix[left delimiter={\{},right delimiter={\}},row sep=2pt,column sep=2pt,inner xsep=0,ampersand replacement=\&]
            {
                \node[inner sep=0]{\tikz[baseline=-0.53ex]{\tileInner{0}{0}{0}{0}},}; \& \node[inner sep=0]{\tikz[baseline=-0.53ex]{\tileInner{1}{0}{1}{0}},}; \\
                \node[inner sep=0]{\tikz[baseline=-0.53ex]{\tileInner{0}{1}{1}{0}},};, \& \node[inner sep=0]{\tikz[baseline=-0.53ex]{\tileInner{1}{1}{0}{1}}\phantom{,}}; \\
            };
        \end{tikzpicture} \\
        \toplefttile &\coloneqq \tikz[baseline=-0.53ex]{\tileInner{2}{2}{5}{5}} \\
        \bottomrighttile  & \coloneqq \tikz[baseline=-0.53ex]{\tileInner{3}{3}{4}{4}}
    \end{align*}
\end{minipage}%
\hspace{.5cm}
\begin{minipage}[m]{0.3\textwidth}
    \centering
    \tilingPictureCounterFixed
    \vspace{-1mm} 
\end{minipage}
\hspace{.45cm}
\hfill
\caption{The instance of a tiling problem $\I_n=(\cT_\increment, \toplefttile, \bottomrighttile, n)$ and the unique valid tiling it enforces.}
\label{fig:problem-instance-and-solution}
\end{figure}

    \section{Emptiness after universal projection is in \expspace}\label{sec:upper-bounds}
    
We now consider algorithmic upper bounds for deciding whether the
language of an automatic relation $R\subseteq (\Sigma^*)^{d+k}$,
given by an \NFA $\cA_R$, is non-empty after a
universal projection step. This yields the upper bound of \Cref{thm:main-theorem}.

Define a homomorphism $h\colon (\Sigma_\#^{d+k})^* \to
  (\Sigma_\#^d)^*$ by
\[
h(a_1,\ldots,a_d,a_{d+1},\ldots,a_{d+k}) \coloneqq
(a_1,\ldots,a_d).
\]
For a given \NFA $\cB$ over $\Sigma_\#^{d+k}$ such that $S\subseteq
(\Sigma^*)^{d+k}$ is automatic via $\cB$ (meaning that $S=\langToRel{\lang{\cB}}$), it is clear that we can
compute in linear time an \NFA $\cB'$ with the same number of states
as $\cB$ such that $L(\cB')=h(\lang{\cB})$. The homomorphism $h$ acts
almost like existential projection, but in general, we do not have
that $\proj[d]{S}$ is automatic via $\cB'$. For instance, suppose that
\[
w = { \begin{bmatrix}
  a \\
  a
\end{bmatrix}\begin{bmatrix}
  b \\
  a
\end{bmatrix}\begin{bmatrix}
  \wSharp\\
  c
\end{bmatrix}\begin{bmatrix}
  \wSharp\\
  a
\end{bmatrix}} \in \lang{\cB}\,.
\]
Then $h(w)=aa\wSharp\wSharp \not\in L_\corrrr$ because of the trailing $\wSharp$ symbols. To remove them, we define
\[
\mathObject{Strip}{L} \coloneqq \set{ w \suchthat \text{there exists
    $v \in (\set[\big]{\wSharp}^d)^*$ such that $wv \in L$} }\,.
\]
It is then the case that $\proj[d]{S}$ is automatic via
$\mathObject{Strip}{\lang{\cB'}} \cap L_\corrrr$. Note that an \NFA for
$\mathObject{Strip}{L}$ can be computed in linear time from an \NFA
for $L$ without changing the set of states by making all states
accepting that can reach a final state via a sequence of
``$\{\wSharp\}^d$'' symbols.

Recall that $\univProj[d]{R} =
\setComplement{\proj[d]{\setComplement{R}}}$, an
automatic presentation of $\univProj[d]{R}$ is thus given by
\[
   \setComplement{
           \bracket{\mathObject{Strip}{
           h \bracket{
               \setComplement{\lang{\cA_R}}
           }
       }
       \cap L_\corrrr}
   }
   \cap L_\corrrr.
\]
Assuming $Q_R$ is the set of states of $\cA_R$, and recalling that
$L_\corrrr \subseteq (\Sigma_\wSharp^d)^*$ is given by an \NFA with
$2^{d+2}$ states, it can easily be checked that the number of
states of an \NFA whose language gives the universal projection of $R$
is bounded by $2^{\big((2^{\size{Q_R} + d + 2}) + d + 2\big)}$.

With those characterisations and estimations at hand, the \expspace
upper bound stated in \Cref{thm:main-theorem} can now easily be
established. 
\begin{prop}
  Let $R$ be an automatic relation given by an \NFA $\cA_R$.
  Deciding whether $\univProj[d]{R}\neq \emptyset$ is in \expspace.
\end{prop}
\begin{proof}
    For an \expspace algorithm, we first construct an NFA
    $\cB=(Q,\Sigma^d_\wSharp,\delta,\set{q_0},F)$ whose language is
    $\big(\mathObject{Strip}[][][\big]{ h(\setComplement{\lang{\cA_R}})}\cap
    L_\corrrr\big)$. We have $\size Q \le 2^{\size{Q_R} + d + 2}$, where $Q_R$
    is the set of states of $\cA_R$, and hence $\cB$ can be constructed
    in exponential space. It remains to show that non-emptiness of
    $\setComplement{\lang{\cB}} \cap L_\corrrr$ can be decided in
    polynomial space.

    Note that we cannot explicitly construct an \NFA for this
    language within polynomial space. Let $\cA_\corrrr=(S,\Sigma_\wSharp^d,\delta_\corrrr,\set{s_0},F_\corrrr)$
    be the \NFA for $L_\corrrr$, we can however non-deterministically
    guess a word in $\setComplement{\lang{\cB}} \cap \lang{\cA_\corrrr}$
    letter by letter as follows. We keep track of a configuration of the
    form $(Q',s) \in 2^Q \times S$, which initially is $(\{ q_0\},
    s_0)$. Then we repeatedly non-deterministically guess some
    $a\in \Sigma^d_\wSharp$ and update $(Q',s)$ to
    $(\set{ \delta(q',a) \suchthat q'\in Q'},\delta_\corrrr(s,a))$ until we reach a configuration
    $(Q',s)$ such that $Q'\cap F=\emptyset$ and $s\in F_\corrrr$. Clearly,
    the word obtained by this sequence of letters is in
    $\setComplement{\lang{\cB}}$ and $\lang{L_\corrrr}$. The overall
    membership in \expspace is then a consequence of Savitch's theorem
    and the observation that the length of the shortest word in
    $\setComplement{\lang{\cB}} \cap L_\corrrr$ is bounded by
    $2^{\big((2^{\size{Q_R} + d + 2}) + d + 2\big)}$.
\end{proof}

    \section{Lower bounds for fragments of \Buchi arithmetic}\label{sec:buchi-lower}
    
In this section, we apply some of the ideas underlying the \expspace lower bound
for universal projection to develop some new lower bounds for
fragments of \Buchi arithmetic with a fixed quantifier alternation
prefix. We explicitly state the lower bounds for \Buchi arithmetic
in base 2, they can, however, easily be generalised to any base $p \ge 2$.
All reductions are easily seen to be \logspace-computable and, again, we do not explicitly argue for that.

    \subsection[A lower bound for the E*A*E* fragment]{A lower bound for the \texorpdfstring{\boldmath$\existsSymbol^*\forallSymbol^*\existsSymbol^*$}{E*A*E*} fragment}
    \label{sec:first-buchi-lower-bound}
    We reduce from the $\tilingProblem$ (\cref{problem:tiling-problem}).
Fix a problem instance $\I = \tuple{\cT, \toplefttile, \bottomrighttile, n}$ and let $W_n\coloneqq2^n$.
The instance $\I$ asks for the existence of a valid tiling of width $W_n$, with $\toplefttile$ in the upper left corner and $\bottomrighttile$ in the bottom right one. W.l.o.g.\ we assume that $\size{\cT} = 2^m$ for some $m \in \N$.

Let $\pbin{k} \in \set{0,1}^+$ denote the binary representation of a number $k \in \N$ in most-significant-digit-first encoding such that $\pbin{k} \in 1 \cdot \set{0,1}^*$ for all $k>0$ and $\pbin{0}\coloneqq 0$.
In what follows, powers of 2 will be used to indicate positions in binary expansions of numbers.
We use the predicate $\pBit{v,x,b}$ defined below for determining whether the bit $b$ occurs in the expansion of some $v$ at a position determined by $x$

\begin{align*}
    \sem{\pBit{}} = {}
    &\set{ (v, x, b) \suchthat \begin{aligned}
                                   &\text{$x = 2^k$ and $v = \sem{\vpre \cdot b \cdot \vsuf}_2$ for some} \\
                                   &\text{$k \in \N$, $\vpre \in \set{0,1}^*$, and $\vsuf \in \set{0,1}^{k-1}$}
    \end{aligned}}\, .
\end{align*}
This predicate can be implemented by an $\existsSymbol^*$
formula of B\"uchi arithmetic as follows:
\begin{align}
    \pBit{v, x, b} \coloneqq {}
    &\exists{\vpre} \exists{y} \exists{\vsuf}. \\
    &\quad\pP{x} \land {} \label{eq:one-bit-lit} \\
    &\quad     (\pV{\vpre, y} \land y > x \lor \vpre = 0) \land {} \label{eq:prefix-has-ones-to-the-left}\\
    &\quad x > \vsuf \land {} \label{eq:sufix-has-ones-to-the-right}\\
    &\quad \bigl((b = 0 \land v = \vpre + \vsuf) \lor (b = 1 \land v = \vpre + x + \vsuf)\bigr)\,.\label{eq:pref-plus-b-plus-suf}
\end{align}
In words, we make sure that $\pbin{x}$ has just one bit set~\eqref{eq:one-bit-lit}.
Then, $\pbin{\vpre}$ is required to have all its digits $1$ to the left of $(\log_2(x)+1)$th least significant bit, or no digits $1$ at all~\eqref{eq:prefix-has-ones-to-the-left};
similarly, $\pbin{\vsuf}$ has to have all digits $1$ to the right~\eqref{eq:sufix-has-ones-to-the-right}.
With all that, we can assert that $\pbin{x} = \pbin{\vpre} \cdot b \cdot \pbin{\vsuf}$~\eqref{eq:pref-plus-b-plus-suf}.

For a given $k\in\N$, by using a ($k+1$)-tuple of numbers $\bar{u}=(u_0,\ldots,u_k)$, we go beyond the binary alphabet $\set{0,1}$ of $\pbin{\argumentDot}$ to the alphabet $\set{0,1,\dots, 2^{k+1}-1}$. The extra variable $u_k$ in $\bar u$ is for technical convenience in order to enable
us to use the formulas defined in this section in the next section as well.
We define an equivalent of $\pBit{}$:
\begin{align*}
    \pRead{\bar{u}, x, a} \coloneqq {}  &\exists{a_0}\exists{a_1} \ldots \exists{a_{k+1}}. \exists{b_1} \ldots \exists{b_{k+1}}.\\
    &{\quad}a_0 = 0 \land a_{k+1} = a \land {}\\
    &{\quad}\smash{\textstyle\bigwedge_{i = 1}^{k+1} \squareBracket{a_i = 2a_{i-1} + b_i \land \pBit{u_{i-1}, x, b_i}}}\,.
\end{align*}
Intuitively, bits at the position determined by $x$ in $\bar{u}$ are interpreted as the binary representation of a number. Note that $\pRead{}$ is also an $\existsSymbol^*$ formula.

Fix some enumeration $\psi\colon \set{0,1,\ldots,2^{m}-1} \to \cT$ of tiles.
We can trivially define quantifier-free formulas $\pTopLeft{}$, $\pBottomRight{}$, $\pMatchH{}$, and $\pMatchV{}$ such that
\begin{align*}
    \sem{\pTopLeft{}} &= \set[\big]{t \in \N \suchthat \psi(t) = \toplefttile}\\
    \sem{\pBottomRight{}} &= \set[\big]{t \in \N \suchthat \psi(t) = \bottomrighttile}\\
    \sem{\pMatchH{}} &= \set{(t, t') \in \N^2 \suchthat \tRight{\psi(t)} = \tLeft{\psi(t')}}\\
    \sem{\pMatchV{}} &= \set{(t, t') \in \N^2 \suchthat \tBottom{\psi(t)} = \tTop{\psi(t')}}\,.
\end{align*}

Our goal is to represent a tiling of height $H$ as a tuple of numbers $\bar u$ such that 
the bottom-right tile is at position one, the top-left at position $W_n\cdot H$, and any tile
$t$ at position $(x,y)$ is at position $W_n\cdot H - (y-1)\cdot W_n - (x - 1)$.
To assert that tiles are properly vertically aligned, we need to compare tiles whose
distance is $W_n$. As stated above, in this section, $W_n=2^n$ and hence of polynomial
bit-length. In principle, we could now straight-forwardly continue with the reduction.
However, in the next section we will have $W_n=2^{2^n}$, which is of exponential bit length.
In order to have a uniform presentation, in this section we give a slightly more sophisticated
reduction that is here easy to follow and can then be enhanced to work for $W_n=2^{2^n}$ in 
the next section.

Our first step is to define a ``ruler''. Intuitively, $\pIsARuler{\bar{u}, F, W_n}$ holds
if for every position $p$ in the interval $[0,F)$, $\pRead_n{\bar u, 2^p, p \bmod W_n}$ holds:
\begin{align*}
    \pIsARuler{\bar{u}, F, W} \coloneqq {}& \pRead_n{\bar{u}, 1, 0} \land \pRead_n{\bar{u}, 2F, 0} \land {}\\
    &\quad\forall{x}. (F > x \geq 1 \land \pP{x}) \limplies {}\\
    &\quad\quad\exists{v'}\exists{v}.\pRead_n{\bar{u}, 2x, v'} \land \pRead_n{\bar{u}, x, v} \land {}\\
    &\quad\quad\quad(v' = v+1 \lor (v' = 0 \land v = W-1))\,.
\end{align*}
Assuming that $\bar u$ satisfies $\pIsARuler{\bar u, F, W}$, we can easily define the set of
positions which are 
multiples of $W_n$ and smaller than $F$:
\[
    \pWidthMul{\bar{u}, x} \coloneqq {} \pRead_n{\bar{u}, x, 0}\,.
\]
\vspace{-\baselineskip}
\begin{fact}
We have $x = 2^{kW_n}$ for some $k \in \N$ if and only if
for some $F$ and $\bar u$
\[
    \pIsARuler{\bar{u}, F, W_n} \land {2F > x \geq 1}  \land \pWidthMul{\bar{u}, x}\,.
\]
\end{fact}

We can also ensure that positions represented by $x$ and $x'$ have distance exactly $W_n$:
\begin{multline*}
    \pWidth{\bar{u}, x, x'} \coloneqq \exists{v}. \pRead{\bar{u}, x, v} \land \pRead{\bar{u}, x', v} \land {}\\
   \forall{y}. (x > y > x') \limplies \neg\pRead{\bar{u}, y, v}\,.
\end{multline*}
\vspace{-\baselineskip}
\begin{fact}
  We have $x = 2^{p+W_n}, y = 2^p$ for some $p$ if and only if for some $F$ and $\bar u$ 
  \[\pIsARuler{\bar{u}, F, W_n} \land {2F > x \geq 1} \land \pWidth{\bar{u}, x, y}\,.
  \]
\end{fact}

We now complete our reduction by defining a formula $\pETiling{}$ of B\"uchi arithmetic 
such that $\sem{\pETiling{}}$ is non-empty
if and only if there exists a valid $(\cT, \toplefttile, \bottomrighttile)$-tiling of width~$W_n$:
\begin{align*}
    \pETiling{} \coloneqq {}
    &\exists{F}.  \exists{u_0} \exists{u_1} \ldots\, \exists{u_{n}}. \exists{w_1} \exists{w_2} \ldots\, \exists{w_m}. \\
    &\quad \pIsARuler{\bar{u}, F, W_n} \land {} \tag{A}\label{eq:is-a-ruler}\\
    &\quad\big[\exists{t}. \pRead_m{\bar{w}, F, t} \land \pTopLeft{t}\big] \land {} \tag{B}\label{eq:top-left-matches}\\
    &\quad\big[\mathrlap{\exists{t}.}\hphantom{\exists{t}.} \pRead_m{\bar{w}, 1, t} \land \pBottomRight{t}\big] \land {} \tag{C}\label{eq:bottom-right-matches}\\
    &\quad \big[\forall{x}. \bracket{F > x \land \pP{x} \land \neg\pWidthMul{\bar{u}, 2x, W_n}} \limplies {}\\
    &\quad\quad\exists{t'}\exists{t}.\pRead_m{\bar{w}, 2x, t'} \land \pRead_m{\bar{w}, x, t} \land \pMatchH{t', t}\big] \land {}\tag{D}\label{eq:horizontal-matching}\\
    &\quad\big[\forall{x}\forall{x'}. \bracket{2F > x > x' \land \pWidth{\bar{u}, x, x', W_n}} \limplies {}\\
    &\quad\quad\exists{t} \exists{t'}.\pRead_m{\bar{w}, x, t} \land \pRead_m{\bar{w}, x', t'} \land \pMatchV{t, t'}\big]\,. \tag{E}\label{eq:vertical-matching}
\end{align*}
Above, in~\eqref{eq:is-a-ruler}, we require that $2F = 2^{k W_n}$ for some $k \in \N$.
We interpret $\bar{u}$ between positions $F$ and $1$ as a word over the alphabet $\set{0,1,\ldots,2^m-1}$. In \crefrange{eq:top-left-matches}{eq:vertical-matching}, we express four properties of a valid tiling: 
\begin{itemize}
    \item in \cref{eq:top-left-matches,eq:bottom-right-matches}: that tiles $\toplefttile$ and $\bottomrighttile$ are on their respective positions, $F$ corresponding to the upper-left corner and $1$ to the lower-right one,
    \item in \eqref{eq:horizontal-matching}: that tiles match horizontally; here, $2x$ and $x$ are two consecutive powers of two, corresponding to horizontally adjacent tiles, unless $\pWidthMul{\bar{u}, 2x, N}$ holds, when $2x$ and $x$ are in two consecutive rows,
    \item in \eqref{eq:vertical-matching}: that tiles match vertically; positions $x, x'$ correspond to vertical neighbours.
\end{itemize}
Keeping in mind that $\phi \limplies \psi \equiv \neg \phi \lor \psi$, we observe that $\pETiling{}$ has a quantifier prefix $\existsSymbol^* \forallSymbol^* \existsSymbol^*$ since 
$\pWidthMul{}$ has an $\existsSymbol^*$ quantifier prefix, $\pWidth{}$ has an $\existsSymbol^* \forallSymbol^*$ quantifier prefix, and $\pIsARuler{}$ has an $\forallSymbol^* \existsSymbol^*$ quantifier prefix.
Finally, observe that the above construction can easily be generalised to \Buchi arithmetic with a $\pV_p{}$ predicate for any $p \ge 2$. This yields:
\begin{prop}    
    Deciding the $\existsSymbol^* \forallSymbol^* \existsSymbol^*$ fragment of \Buchi arithmetic is \expspace-hard.
\end{prop}

    \subsection[A lower bound for the E*A*E*A* fragment]{A lower bound for the \texorpdfstring{\boldmath$\existsSymbol^*\forallSymbol^*\existsSymbol^*\forallSymbol^*$}{E*A*E*A*} fragment}
    \label{sec:second-buchi-lower-bound}
    Building upon the reduction given in the previous section, we now give a \expexpspace lower bound
for the $\existsSymbol^* \forallSymbol^* \existsSymbol^* \forallSymbol^*$ fragment of \Buchi arithmetic. 
We reduce from a modified variant of $\tilingProblem$ (\cref{problem:tiling-problem}), where we require the tiling width $W_n' \coloneqq 2^{W_n}= 2^{2^{n}}$; this is trivially a \expexpspace-hard problem since it allows for simulating a Turing machine with
a tape with a doubly exponential number of cells.
The reduction has the same general structure as that in \cref{sec:first-buchi-lower-bound}, and the only difference lies in the definition of formulas $\pWidthMul{}$, $\pWidth{}$, and $\pIsARuler{}$, which we substitute with $\pWidthMul^{\prime}{}$, $\pWidth^{\prime}{}$, and $\pIsARuler^{\prime}{}$, respectively, yielding a new formula $\pETiling^{\prime}{}$.
All three sub-formulas are given below. Observe that now $\pETiling^{\prime}{}$ has the quantifier pattern $\existsSymbol^* \forallSymbol^* \existsSymbol^* \forallSymbol^*$ if we provide:
\begin{itemize}
    \item $\pWidthMul^{\prime}{}$ with an $\existsSymbol^*$ quantifier prefix,
    \item $\pWidth^{\prime}{}$ with an $\existsSymbol^* \forallSymbol^* \existsSymbol^*$ quantifier prefix,
    \item $\pIsARuler^{\prime}{}$ with an $\existsSymbol^* \forallSymbol^* \existsSymbol^* \forallSymbol^*$ quantifier prefix.
\end{itemize}

We now give the definition of $\pIsARuler^{\prime}{}$. Intuitively, the ruler in this case takes the shape of a comb from \cref{subsec:word-encoding-of-tiling}:
\begin{align*}
    \pIsARuler^{\prime}{\bar{u}, F, W} \coloneqq {} & \pRead_n{\bar{u}, 2F, W-1} \land \pRead_n{\bar{u}, 1, W-1} \land {}\\
    &\forall{x}\forall{x'}\forall{v}\forall{v'}.\\
    &\,\bracket{2F \geq x > x' \geq 1 \land \pRead_n{\bar{u}, x, v} \land \pRead_n{\bar{u}, x', v'}} \limplies {}\\
    &\,\;\big[\exists{y}\exists{v_y}. x > y > x' \land \pRead_n{\bar{u}, y, v_y} \land v_y \geq \min\set{v, v'}\big] \lor {}\\
    &\,\;\big[\exists{y}\exists{v_y}. x > y > x' \land \pRead_n{\bar{u}, y, v_y} \land v_y = \min\set{v, v'}-1 \land {}\\
    &\,\;\forall{y'}\forall{v'_y}. x < y' < x' \land y' \neq y \land \pRead_n{\bar{u}, y', v'_y} \limplies v'_y < v_y\big] \lor {}\\
    &\,\;\squareBracket{v = 0 \lor v' = 0}
\end{align*}
Let $c$ be $\comb{W}$ without the first symbol, i.e., $c$ has length $W_n'$.
Intuitively, $\pIsARuler^{\prime}{}$ defines a periodic sequence $W_n, c^+$ on $\bar{u}$ observed through the lens of $\pRead_n{\bar{u}, x_i, \argumentDot}$ for consecutive positions $x_i$. We achieve this by requiring 
precisely one occurrence of the number ${\min\set{k, k'}-1}$ between numbers $k, k' \in \N$, unless $\min\set{k, k'} = 0$ or there is some $k'' \geq \min\set{k,k'}$ between them.
We can now simply put
\begin{align*}
    \pWidthMul^{\prime}{\bar{u}, x, W} \coloneqq {} & \pRead_n{\bar{u}, x, W}\,.
\end{align*}
Indeed, as the length of $c$ is precisely $W'_n$ and $c$ contains only one occurrence of $W_n$ at the rightmost position, positions marked with $W_n$ correspond to $2^{\N W'_n} \cap [1,2F]$.

To verify that the distance between $x$ and $y$ is exactly $W'_n$, we first make sure that $x$ and $y$ fall into two consecutive rows using $\pWidthMul^{\prime}{}$. If so, they are in the same position within a row, if the corresponding bits of the binary counter at $x$ and $y$ implicitly determined by $\bar{u}$ in line with \cref{obs:crucial-observation} are equal. This yields an $\existsSymbol^* \forallSymbol^* \existsSymbol^*$ formula:
\begin{align*}
    &\pWidth^{\prime}{\bar{u}, x, y, W} \coloneqq {}\\
    &\qquad\exists{p}\exists{q}\exists{r}.p > q > r \land {}\tag{A.1}\label{eq:segments-pqr-exist-a}\\
    &\qquad\pWidthMul^{\prime}{\bar{u}, p, W} \land \pWidthMul^{\prime}{\bar{u}, q, W} \land \pWidthMul^{\prime}{\bar{u}, r, W} \land {}\tag{A.2}\label{eq:segments-pqr-exist-b}\\
    &\qquad\squareBracket{\forall{s} (p > s > r \land \pWidthMul^{\prime}{\bar{u}, s, W}) \limplies s = q} \land {}\tag{A.3}\label{eq:segments-are-close}\\
    &\qquad p > x \geq q > y \geq r \land {}\tag{A.4}\label{eq:segments-contain-xy}\\
    &\qquad\forall{B}. (0 \leq B < W) \limplies {}\tag{B}\label{eq:for-all-constants}\\
    &\qquad\quad{\big[}\bracket{\forall{x'}. (p > x' \geq x) \limplies \neg \pRead{\bar{u}, x', B}}\land {}\tag{C.1}\label{eq:either-both-segments}\\
    &\qquad\quad\hphantom{\big[}\bracket{\forall{y'}. (q > y' \geq y) \limplies \neg \pRead{\bar{u}, y', B}}\big] \lor {}\tag{C.2}\label{eq:either-both-segments-do-not-see-B}\\
    &\qquad\quad\forall{x'}\forall{y'}. (\pRead_n{\bar{u}, x', B} \land \pRead_n{\bar{u}, y', B}) \limplies {}\tag{D}\label{eq:or-if-they-both-see-B}\\
    &\qquad\quad\quad\squareBracket{\exists{x''}. x' > x'' \geq x \land \pRead{\bar{u}, x'', B}} \quad \lor {}\tag{E.1}\label{eq:xx-is-chosen-suboptimally}\\
    &\qquad\quad\quad\squareBracket{\exists{y''}. y' > y'' \geq y \land \pRead{\bar{u}, y'', B}} \quad \lor {}\tag{E.2}\label{eq:yy-is-chosen-suboptimally}\\
    &\qquad\quad\quad\bigl[
    \bracket{\forall{x''}. x' > x'' \geq x \land \pP{x} \limplies \exists{v} \pRead{\bar{u}, x'', v} \land v < B} \land {}\tag{F.1}\label{eq:both-see-B}\\
    &\qquad\quad\quad\hphantom{\bigl[}
    \bracket{\forall{y''}. y' > y'' \geq y \land \pP{y} \limplies \exists{v} \pRead{\bar{u}, y'', v} \land v < B}\bigr] \quad \lor {}\tag{F.2}\label{eq:both-see-B-b}\\
    &\qquad\quad\quad\bigl[
    \bracket{\exists{x''}\exists{v}. x' > x'' \geq x \land \pRead{\bar{u}, x'', v} \land v > B} \land {}\tag{G.1}\label{eq:both-do-not-see-B}\\
    &\qquad\quad\quad\hphantom{\bigl[}
    \bracket{\exists{y''}\exists{v}. y' > y'' \geq y \land \pRead{\bar{u}, y'', v} \land v > B}\bigr]\,.\tag{G.2}\label{eq:both-do-not-see-B-b}
\end{align*}
Intuitively, we assert that $p,q,r$ delimit two consecutive rows in which $x$ and $y$ fall; cf.\ \crefrange{eq:segments-pqr-exist-a}{eq:segments-contain-xy}. 
The remainder of the formula, in \crefrange{eq:either-both-segments}{eq:both-do-not-see-B-b}, asserts that the $B$th bit of the (virtual) $2^n$-bit counter is equal for $x$ and $y$; with \cref{eq:for-all-constants}, we perform that test for every $B$.
In \cref{eq:either-both-segments,eq:either-both-segments-do-not-see-B} we test whether $B$ is absent in segments between $p$ and $x$, and between $q$ and $y$. If so, the $B$th bit is $0$ for both $x$ and $y$, and thus the $B$th bit equality test is successful. If it is the case only for one segment, $x$ and $y$ are not vertical neighbours, and the formula is false.
If in turn both segments contain $B$, we catch the cases where $x'$ or $y'$ are not the rightmost occurrences of $B$ in \cref{eq:xx-is-chosen-suboptimally,eq:yy-is-chosen-suboptimally}. If both $x'$ and $y'$ are rightmost occurrences, we require either both segments to consist only of values less than $B$; cf.\  \cref{eq:both-see-B,eq:both-see-B-b}, or both to contain a value greater than $B$; cf.\ \cref{eq:both-do-not-see-B,eq:both-do-not-see-B-b}. This makes sure that the value of the $B$th bit is equal for $x$ and $y$, and therefore, they are vertical neighbours. 

As before, the above construction can be generalised to \Buchi arithmetic with a $\pV_p{}$ predicate for any $p \ge 2$.
This yields:
\begin{prop}    
    \label{prop:BA-fragment-EXPSPACE-hard}
    Deciding the $\existsSymbol^* \forallSymbol^* \existsSymbol^* \forallSymbol^*$ fragment of \Buchi arithmetic is \expexpspace-hard.
\end{prop}

    \section{Conclusion}
    
In this paper, we studied the computational complexity of eliminating
universal quantifiers in automatic structures. We showed that, in
general, this is a computationally challenging problem whose
associated decision problem is \expspace-complete.
Our result further reinforces the intuition already stemming from~\cite{Kus09a} that, in general, the alternation of quantifiers requires ``complex'' automata. We also
used the technical construction underlying the \expspace lower bound to obtain new lower bounds for deciding
formulas of B\"uchi arithmetic with a fixed quantifier
alternation prefix.

It would be interesting to understand whether it is possible to
identify natural sufficient conditions on regular languages for which
a universal projection step does not result in a doubly exponential
blow-up and only leads to, e.g., polynomial or singly exponential
growth. Results of this kind have been obtained in model-theoretic terms for
structures of bounded degree~\cite{KL11,GH12}, but we are not aware of
a systematic study of questions of this kind on the level of regular
languages.

Finally, the lower bounds established in \Cref{sec:buchi-lower} for deciding formulas of B\"uchi arithmetic with fixed quantifier alternation prefixes are not tight. It can be seen that the upper
bounds, which can be derived from \Cref{sec:upper-bounds} for those
classes of formulas, are
off by one exponential. We leave it as an open problem
of this paper to fully characterise the complexity
of deciding B\"uchi arithmetic with an arbitrary but
fixed quantifier alternation prefix. A first step could be to
uniformise the constructions in \Cref{sec:buchi-lower} to enable
showing $k$-\expspace-hardness for any $k\ge 1$. At present, it
is not obvious to us how this could be achieved.

    \section*{Acknowledgements}
    We would like to thank the anonymous CONCUR'23 and LMCS reviewers for their comments and suggestions, which helped us to improve the presentation of this paper. We are grateful to Alessio Mansutti for spotting a bug in \Cref{sec:size-of-automaton}. This work is part of a project that has received funding from the European Research
Council (ERC) under the European Union’s Horizon 2020 research and innovation programme (Grant agreement No.\ 852769, ARiAT).
    
    \bibliographystyle{alphaurl}
    \bibliography{bibliography}
\end{document}